\documentclass[11pt]{article}
\usepackage{color}
\usepackage{amsmath}
\usepackage{mathrsfs}
\usepackage{amssymb}
\usepackage{amsthm}
\usepackage{epstopdf}
\usepackage{graphicx}
\usepackage{mathtools}
\usepackage[usenames,dvipsnames, table]{xcolor}
\usepackage[hang]{caption}
\usepackage{hyperref}

\usepackage{oldgerm}  
\usepackage[yyyymmdd,hhmmss]{datetime}

\usepackage{mathabx} 
\def\bee{\begin{enumerate}}\def\eee{\end{enumerate}}
\def\bei{\begin{itemize}}\def\eei{\end{itemize}}
\newcommand{\nco}{\newcommand}
\def\R{\mathbb{R}}
\nco{\red}{\color{red}}
\nco{\blue}{\color{blue}}
\nco{\cyan}{\color{cyan}}
\nco{\brown}{\color{Magenta}}

\nco{\magenta}{\color{magenta}}

\nco{\violet}{\color{violet}}
\nco{\olive}{\color{Emerald}}
\nco{\orange}{\color{orange}}
\nco{\redend}{\normalcolor}
\nco{\blueend}{\normalcolor}
\def\inv#1{\frac{1}{#1}}
\def\tr{{\rm tr}\,}
\def\sign{{\rm sign}\,}

\def\ommit#1{{}}
\def\({\left(}
\def\){\right)}
\def\ie{{\it i.e.,\/}\ }
\def\ie{{\rm i.e.,\/}\ }

\def\smat#1{\mbox{\footnotesize{\mbox{$\begin{pmatrix}#1\end{pmatrix}$}}}}
\nco{\rnc}{\renewcommand}
\rnc{\title}[1]{{\Large\bf\mbox{}\\\medskip#1\bigskip\medskip\\}}
\rnc{\author}[1]{{\large #1\smallskip\\}}
\nco{\address}[1]{{\em #1\medskip\\}}

\def\diag{{\rm diag \,}}
\def\ii{\mathrm{i\,}}
\nco{\bun}{{\bf 1}}
\def\be{\begin{equation}}\def\ee{\end{equation}}
\def\bea{\begin{eqnarray}}\def\eea{\end{eqnarray}}
\def\bee{\begin{enumerate}}\def\eee{\end{enumerate}}
\def\bei{\begin{itemize}}\def\eei{\end{itemize}}
\def\oh{\frac{1}{2}}
\def\ommit#1{{}}

\def\SU{{\rm SU}}\def\U{{\rm U}}\def\SO{{\rm SO}}
\def\inv#1{\frac{1}{#1}}
\def\mult{{\rm mult}}
\def\tr{{\rm tr\, }} 
\def\smat#1{\mbox{\footnotesize{\mbox{$\begin{pmatrix}#1\end{pmatrix}$}}}}

\def\eq=#1{\buildrel #1 \over{=}}

\def\CC{{\mathcal C}}
\def\CM{{\mathcal M}}
\def\CH{{\mathcal H}}     \def\CJ{{\mathcal J}}   \def\CO{{\mathcal O}}

\def\Ga{\pmb{\alpha}}

\def\rvol{\mathrm{Vol}_\mathrm{rel}}
\def\c{c}  \def\c{r}    
\def\t{u}    

\def\gog{{\mathfrak g}}
\def\diag{{\rm diag \,}}
\def\ii{\mathrm{i\,}}
\def\N{\mathbb{N}}
\def\Z{\mathbb{Z}}

\def\R{\mathbb{R}}

\def\T{\mathbb{T}}

\newtheorem{proposition}{Proposition}

\def\hDelta{\hat\Delta}

\begin{document}

\begin{titlepage}  
\begin{center}
\title{On  Horn's Problem and its Volume Function}
\medskip
\author{Robert Coquereaux} 
\address{Aix Marseille Univ, Universit\'e de Toulon, CNRS, CPT, Marseille, France\\
Centre de Physique Th\'eorique}
\bigskip\medskip
\medskip
\author{Colin McSwiggen${}^{\dagger\ddagger}$ and Jean-Bernard Zuber${}^\ddagger$\footnote{zuber@lpthe.jussieu.fr   \qquad ORCID 0000-0001-6907-6538}}
\address{
$\dagger$ Brown University, Division of Applied Mathematics, Providence, USA \\
 \ddag \ Sorbonne Universit\'e,  UMR 7589, LPTHE, F-75005,  Paris, France\\ \& CNRS, UMR 7589, LPTHE, F-75005, Paris, France
 }

\bigskip\bigskip

\bigskip\bigskip
\begin{abstract}
We consider an extended version of Horn's problem: 
given two orbits $ {\mathcal O}_\alpha$ and $ {\mathcal O}_\beta$ of a  linear representation of a compact Lie group, let $A\in  {\mathcal O}_\alpha$, $B\in  {\mathcal O}_\beta$ be independent and invariantly distributed random elements of the two orbits.  The problem is to describe the probability distribution of the orbit of the sum $A+B$.
We study in particular the familiar case of coadjoint orbits, and also the orbits of self-adjoint 
real, complex and quaternionic matrices under the conjugation actions of $\SO(n)$, $\SU(n)$ and $\mathrm{USp}(n)$ respectively.
The probability density can be expressed in terms of a function that we call the volume function. In  this paper, (i) we relate this function to the symplectic or Riemannian geometry of the orbits, depending on the case; (ii) 
we discuss  its non-analyticities and possible vanishing; (iii) in the coadjoint case, we study its relation to tensor product multiplicities (generalized Littlewood--Richardson 
coefficients) and show that it computes the volume of a family of convex polytopes introduced by Berenstein and Zelevinsky. These considerations are illustrated by a detailed study of the volume function for the
coadjoint orbits of $B_2=\frak{so}(5)$.
\end{abstract}
\end{center}

 \end{titlepage}
 
  
  \section{Introduction}
 Horn's problem is the following question.  Given $n$-by-$n$ Hermitian matrices $A$ and $B$ with known eigenvalues $\alpha_1 \ge \hdots \ge \alpha_n$ and $\beta_1 \ge \hdots \ge \beta_n$, what can be said about the eigenvalues $\gamma_1 \ge \hdots \ge \gamma_n$ of their sum $C = A+B$?  After decades of work by many mathematicians, the answer to this question is now well known
\cite{Ho, Kly, KT00} 

There is an extension of Horn's problem that is both more general and more quantitative.  Let $V$ be a representation of a compact Lie group $G$.  To each $G$-orbit $\CO \subset V$ we associate the {\it orbital measure at $\CO$}, which is the unique $G$-invariant probability measure on $V$ that is concentrated on $\CO$.  The {\it orbit space} is the topological quotient $V/G$, in which each point corresponds to a $G$-orbit.  If $\CO_\alpha$ and $\CO_\beta$ are two such orbits and we choose $A \in \CO_\alpha$ and $B \in \CO_\beta$ independently at random from their respective orbital measures, the sum $C=A+B$ will lie in a random orbit $\CO_\gamma$.  For each pair $(\CO_\alpha, \CO_\beta)$, we thus obtain a probability measure on the orbit space, called the {\it Horn probability measure}.  Concretely, $C \in V$ is distributed according to the convolution of the orbital measures at $\CO_\alpha$ and $\CO_\beta$, and the Horn probability measure is the pushforward of this convolution by the quotient map $V \to V/G$.  The extended problem is then to give an explicit description of  the Horn probability measure, whereas the original Horn's problem is to describe only the support of this measure in the specific case where $V$ is the coadjoint representation of $\U(n)$.

In this paper, we study the extended Horn's problem in two families of cases that are of special interest:
\begin{enumerate}
\item \textbf{Coadjoint representations:} $G$ is an arbitrary compact, connected, semisimple Lie group acting by the coadjoint representation on the dual of its Lie algebra $\gog$.
\item \textbf{Spaces of self-adjoint matrices:} $G$ is one of the classical groups $\SO(n),\ \SU(n)$ or $\mathrm{USp}(n)$ acting by conjugation on, respectively, real symmetric, complex Hermitian or quaternionic self-dual matrices. Following convention, in these cases we study the distribution of the {\it sorted eigenvalues} of $A+B$ rather than its orbit.\footnote{ In most cases this distinction is immaterial because the spectrum of $A+B$ uniquely determines its orbit.  The only exceptions are the even special orthogonal groups $G = SO(2n)$, in which case $\diag(x_1, \hdots, x_n)$ and $\diag(x_{w(1)}, \hdots, x_{w(n)})$ may lie in different orbits when $w \in S_n$ is an odd permutation.}
\end{enumerate}
We will refer to these respectively as the ``coadjoint case'' and the ``self-adjoint case.''

The main object of our study will be a function $\CJ$, called the {\it volume function},  which can be computed from the density of the Horn probability measure (and vice versa). This function encodes various kinds of geometric information about the orbits, and in the coadjoint case it additionally encodes combinatorial information related to tensor product multiplicities of irreducible representations of $\gog$.  We discuss the singular and vanishing loci of $\CJ$, its relationship to the Riemannian geometry of the orbits as submanifolds of $V$, and (in the coadjoint case) its interpretation as both a symplectic volume and the volume of a convex polytope, as well as identities that relate $\CJ$ to the tensor product multiplicities of $\gog$.  Finally, we carry out a detailed case study of the coadjoint case for $\gog = \mathfrak{so}(5)$.

 This paper discusses several constructions related to the extended Horn's problem and to tensor product multiplicities, including a number of previously known results that we recall for the sake of completeness.  However, to the authors' knowledge, Propositions \ref{prop1} through \ref{prop:J-positive-interior}, equation (\ref{eqn:J-fiber-int}), and the conjectured expression (\ref{last-formula}) either are new or extend in various ways several results previously obtained by two of us in \cite{Z1, CZ1, CZ2}.  Proposition \ref{prop:J-equals-BZvol} is a particular instance of a more general phenomenon whereby the volumes of certain symplectic orbifolds equal the volumes of polytopes whose integer points count representation multiplicities \cite{Heck, GLS}.  As far as we have been able to determine however, our proof is novel and the precise statement has not appeared previously.
 
In the following two subsections we define $\CJ$ for the coadjoint and self-adjoint cases.  To avoid overloading notation we give separate definitions in the two cases, but the concepts are analogous.

\subsection{$\CJ$ in the coadjoint case}

We first develop some preliminaries related to Horn's problem.  Let $\langle \cdot, \cdot \rangle$ be the $G$-invariant inner product given by $-1$ times the Killing form.  We identify $\gog \cong \gog^*$ using the inner product and we then identify the orbit space with the dominant Weyl chamber $\CC_+$ of a Cartan subalgebra $\mathfrak{t} \subset \gog$, so that the quotient map $\mathfrak{g}^* \to \mathfrak{g}^*/G$ sends each orbit to its unique representative in $\CC_+$.  We identify functions on $\CC_+$ with their unique $G$-invariant extension to $\gog$.  The Jacobian of the quotient map $\gog^* \to \gog^*/G$ is equal to $\kappa_\gog \Delta_\gog(x)^2$, where $\Delta_\gog(x):= \prod_{\Ga >0} \langle \Ga,x\rangle$ is the product of the positive roots of $\gog$ and $\kappa_\gog$ is a numerical coefficient. 
  For the classical Lie algebras, $\kappa_\gog$
  may be determined by computing in two different ways a Gaussian integral over $\gog$ and making 
 use of the Macdonald--Opdam integral \cite{M-O}, giving 
 \be\label{kappag} \kappa_\gog = \frac{(2\pi)^{N_r}}{\Delta_\gog(\rho)}=  \frac{(2\pi)^{N_r}}{\prod_{i=1}^r \ell_i !} \times K\ee
 in terms of the Weyl vector $\rho = \frac{1}{2} \sum_{\Ga > 0} \Ga$, the rank $r$, the number $N_r$ of positive roots, and the Coxeter exponents $\ell_i$ of $\gog$. The coefficient $K=\prod_{\alpha>0} \frac{\langle \pmb{\theta},\pmb{\theta}\rangle}{\langle\Ga,\Ga\rangle}$,
 where $\pmb{\theta}$ is any long root, 
 equals 1 for simply laced algebras, while for the non-simply laced cases it takes the values  
 $K= 2^r$ for $ B_r \ (r> 1)$, 
$K = 2^{r(r-1)}$ for $C_r \ (r> 1)$,  $K = 2^{12}$ for $F_4$ and $K = 3^3$ for $G_2$ \cite{Coq}.

 An element $x \in \mathfrak{t}$ is said to be {\it regular} if $\Delta_\gog(x) \neq 0$. This term should not be confused with the more general notion of a {\it regular value} of a differentiable map, which we will also use frequently. For $x$ regular, the Jacobian of the quotient map is equal to the Riemannian volume of the orbit $\CO_x$ with respect to the metric induced by the inner product.   (Note that this Riemannian volume differs from the symplectic volume discussed below.)

The Horn probability measure is supported on a convex polytope $\mathscr{H}_{\alpha \beta} \subset \CC_+$, called the {\it Horn polytope}, and is absolutely continuous with respect to the induced Lebesgue measure on $\mathscr{H}_{\alpha \beta}$.   We assume in what follows that $\alpha$ and $\beta$ are regular, in which case $\dim \mathscr{H}_{\alpha \beta} = r$ and the measure has a global density on $\CC_+$.

  Our discussion of the Horn probability measure will make ubiquitous use of the {\it orbital integral} (also called the Harish-Chandra orbital function), defined for $\alpha, x \in \CC_+$ as
  \be\label{orb-int}  \CH(\alpha, \ii x)  =\int_G dg\ e^{\ii \langle g \cdot \alpha, x \rangle},
  \ee
where $dg$ is normalized Haar measure. Considered as a $G$-invariant function {of $x \in \gog$, $\CH(\alpha, \ii x)$ } is the Fourier transform of the orbital measure at $\CO_\alpha$, {so that the characteristic function of the convolution of orbital measures at $\CO_\alpha$ and $\CO_\beta$ is the product $\CH(\alpha, \ii x) \CH(\beta, \ii x)$.}

The probability density function (PDF) of the Horn probability measure can be written in terms of orbital integrals {by taking the inverse Fourier transform of this characteristic function, rewriting it as a function of $\gamma \in \CC_+$, and multiplying by the Jacobian $\kappa_\gog \Delta_\gog(\gamma)^2$ to account for the quotient map.  The resulting expression for the PDF is}
  \bea \nonumber
  p(\gamma | \alpha, \beta)&=&  
 \frac{\kappa_\gog^2\ \Delta_\gog(\gamma)^2}{(2\pi)^{\dim \gog} |W|}   \int_\mathfrak{t} dx\ \Delta_\gog(x)^2 \CH(\alpha, \ii x) \CH(\beta, \ii x) \big(\CH(\gamma, \ii x)\big)^* \\
&=&
\label{PDF-coadjoint}  
 \frac{1}{(2\pi)^{r} |W| } \left( \frac{\Delta_\gog(\gamma)}{\Delta_\gog(\rho)} \right)^2  \int_\mathfrak{t} dx\ \Delta_\gog(x)^2 \CH(\alpha, \ii x) \CH(\beta, \ii x) \big(\CH(\gamma, \ii x)\big)^*
  \eea
 where   $dx$ is the Lebesgue measure on $\mathfrak{t}$ associated to the inner product $\langle \cdot, \cdot \rangle$
and $|W|$ is the order of the Weyl group. Similar formulae for the convolution of orbital measures have appeared in \cite{DOW}. 

 The volume function $\CJ(\alpha, \beta; \gamma)$ is then defined as
\bea \nonumber
\CJ(\alpha, \beta; \gamma) &:=& 
{\frac{\Delta_\gog(\alpha) \Delta_\gog(\beta)}{\Delta_\gog(\gamma) \Delta_\gog(\rho)}}\ p(\gamma | \alpha, \beta) \\
\label{eqn:J-def-coadjoint}
&=& \frac{\Delta_\gog(\alpha) \Delta_\gog(\beta) \Delta_\gog(\gamma)}{(2 \pi)^r\, |W|\, \Delta_\gog(\rho)^3} \int_\mathfrak{t} dx\ \Delta_\gog(x)^2 \CH(\alpha, \ii x) \CH(\beta, \ii x) \big(\CH(\gamma, \ii x)\big)^*.
\eea
For fixed $\alpha$ and $\beta$ it is a piecewise polynomial function of $\gamma$.  The last line of (\ref{eqn:J-def-coadjoint}) can be used to define $\CJ$ without assuming that $\alpha$ or $\beta$ is regular, though one finds that $\CJ$ vanishes for non-regular arguments.  Note that although we take $\gamma \in \CC_+$ above, the expression (\ref{eqn:J-def-coadjoint}) extends naturally to a function on $\mathfrak{t}$ that is skew-invariant under the action of $W$.

The volume function is the central concern of this paper, and it admits several interpretations.  To begin with, $\CJ$ computes the symplectic (Liouville) volume of a family of symplectic orbifolds parametrized by the triple $(\alpha, \beta, \gamma)$.  Here we take $\alpha, \beta$ regular and $\gamma$ on the interior of a polynomial domain of $\CJ$. Coadjoint orbits admit a canonical $G$-invariant symplectic form, the Kostant--Kirillov--Souriau form \cite{Kir}, for which the inclusion into $\gog^*$ is a moment map.  For $x \in \mathfrak{t}^*$ regular, the Liouville volume of the orbit $\CO_x$ is then equal to $\Delta_\gog(x) / \Delta_\gog(\rho)$.  (See e.g. \cite{McS} sects. 4.2 and 4.3 for a derivation of this well-known fact.)  The product of orbits $\CO_\alpha \times \CO_\beta \times \CO_{-\gamma}$ carries a diagonal $G$-action with moment map $\phi: (A, B, -C) \mapsto A + B - C$.  Let $\mu$ be the Liouville volume measure on $\CO_\alpha \times \CO_\beta \times \CO_{-\gamma}$.  The pushforward $\phi_* \mu$ is the Borel measure on $\gog^*$ defined by $\phi_*\mu (U) = \mu(\phi^{-1}(U))$, $U \subset \gog^*$.  This measure is equal to the convolution of the three orbital measures times the volumes of the orbits, so that it has a density (Radon-Nikodym derivative) given by
 {\bea \nonumber
 h_{\alpha \beta}^\gamma(z) \!\!\! &=& \!\!\! \frac{1}{(2\pi)^{\dim \gog}} \frac{\Delta_\gog(\alpha) \Delta_\gog(\beta) \Delta_\gog(\gamma)}{\Delta_\gog(\rho)^3} \int_\mathfrak{\gog} dx\ \CH(\alpha, \ii x) \CH(\beta, \ii x) \CH(-\gamma, \ii x)\, e^{-\ii \langle x, z \rangle} \\
 \label{eqn:DH-density}
 &=& \!\!\! \frac{\kappa_\gog}{(2\pi)^{\dim \gog} |W|} \frac{\Delta_\gog(\alpha) \Delta_\gog(\beta) \Delta_\gog(\gamma)}{\Delta_\gog(\rho)^3} \int_\mathfrak{t} dx\ \Delta_\gog(x)^2 \CH(\alpha, \ii x) \CH(\beta, \ii x) \big[\CH(\gamma, \ii x)\CH(z, \ii x)\big]^*
 \eea 
 with respect to Lebesgue measure $dz$ on $\gog^*$.}  Thus we find that
 \be \CJ(\alpha, \beta; \gamma) = (2\pi)^{N_r} \Delta_\gog(\rho) \, h_{\alpha \beta}^\gamma(0). \ee
By the theory of Duistermaat-Heckman measures \cite{DH}, $h_{\alpha \beta}^\gamma(0)$ equals the symplectic volume of $\phi^{-1}(0) / G$, the Hamiltonian reduction of $\CO_\alpha \times \CO_\beta \times \CO_{-\gamma}$ at level 0, which is a symplectic orbifold. (In some cases the quotient $\phi^{-1}(0) / G$ is smooth, so that it is a genuine symplectic manifold, but in general it may have singular points.) For a more detailed discussion of these symplectic quotients in the context of the classical Horn's problem, we refer the reader to \cite{Knut}.

The volume function is also related to tensor product multiplicities (generalized Littlewood-Richardson coefficients) in representation theory.  Let $V_\lambda, V_\mu, V_\nu$ be irreducible representations of $\mathfrak{g}$ with highest weights $\lambda, \mu, \nu$, which we assume to be a {\it compatible triple}, meaning that $\lambda + \mu - \nu$ belongs to the root lattice.  Then, in the language of geometric quantization, $\CJ(\lambda, \mu; \nu)$ is a ``semiclassical approximation'' for the tensor product multiplicity $C_{\lambda \mu}^\nu := \dim \mathrm{Hom}_\gog ( V_\lambda \otimes V_\mu \to V_\nu). $
Indeed, a connection with representation theory is already apparent in (\ref{eqn:J-def-coadjoint}).  Let a prime denote the Weyl shift of a weight: $\lambda' = \lambda + \rho$.  By the Weyl dimension formula, $\dim V_\lambda = \Delta_\gog(\lambda')/\Delta_\gog(\rho),$ so that we have
\be \label{eqn:J-Weyl-dimension}
\CJ(\lambda', \mu'; \nu') = \frac{\dim V_\lambda \dim V_\mu \dim V_\nu}{(2 \pi)^r |W|} \int_\mathfrak{t} dx\ \Delta_\gog(x)^2 \CH(\lambda', \ii x) \CH(\mu', \ii x) \big(\CH(\nu', \ii x)\big)^*.
\ee
We explore the representation-theoretic significance of $\CJ$ more thoroughly below in sect. \ref{sectJ-LR}, where we derive two explicit relations expressing $\CJ$ in terms of the multiplicities $C_{\lambda \mu}^\nu$.  In sect. \ref{sectVol} we provide yet another perspective on the relationship between the volume function and tensor product multiplicities, by showing that that $\CJ(\lambda, \mu; \nu)$ is proportional to the Euclidean volume of a polytope whose number of integer points is equal to $C_{\lambda \mu}^\nu$.


\subsection{$\CJ$ in the self-adjoint case}
\label{sec:intro_selfadjoint}

The volume function in the self-adjoint case does not offer as rich an array of geometric interpretations as in the coadjoint case, largely due to the fact that, in general, the orbits do not carry a natural symplectic structure.  The exception is when $G = \SU(n)$, in which case the coadjoint representation is equivalent to the action by conjugation on traceless Hermitian matrices, so that for $\alpha, \beta$ traceless the coadjoint and self-adjoint cases coincide.  However, for real symmetric or quaternionic self-dual matrices we are not dealing with a coadjoint representation, and the interpretations of $\CJ$ as the volume of a symplectic manifold or a polytope do not apply.  Nevertheless, $\CJ$ still encodes substantial geometric information.

 We first observe that a translation $A \mapsto A + aI,$ $B \mapsto B + bI$ ($a, b \in \R$) merely translates the distribution of each $\gamma_i$ by $a + b$.  Therefore, up to a translation of its support, the Horn probability measure depends only on the trace-free parts of $\alpha$ and $\beta$. Accordingly, in what follows, we assume without loss of generality that $A$ and $B$ are traceless.

Let $G = \SO(n),\ \SU(n),$ or $\mathrm{USp}(n)$.  We label these three cases by a parameter\footnotemark \footnotetext{In the language of $\beta$-ensembles appearing in random matrix theory, our $\theta$ is equal to $\beta/2$.  We opted for this notation, which is more common in symmetric function theory, to avoid overloading the symbol $\beta$.} $\theta$ that respectively equals $1/2$, $1$ or $2$.  Let $\CM_{\theta,n}$ be respectively the space of $n$-by-$n$ real symmetric, complex Hermitian, or quaternionic  self-dual matrices.  Then $G$ acts on $\CM_{\theta,n}$ by conjugation.  Let $\CM_{\theta,n}^0$ be the subspace of traceless matrices in $\CM_{\theta,n}$.
  We define a $G$-invariant inner product $\langle \cdot, \cdot \rangle$ on $\CM_{\theta,n}$ using the trace form $\langle A,B \rangle = \tr(AB)$.  The space of spectra of matrices in $\CM_{\theta,n}^0$ is naturally identified with the space of real diagonal matrices $\diag(x_1, \hdots, x_n)$ such that $x_1 \ge \hdots \ge x_n$ and $\sum x_i=0$, which we also denote $\CC_+$.  We identify the space of all real traceless diagonal matrices with $\R^{n-1}$, and  we identify functions on $\CC_+$ with their symmetric (in the $x_i$'s) extensions to $\R^{n-1}$.

 The Jacobian of the diagonalization map is equal to $\kappa_\theta |\Delta(x)|^{2\theta}$, where $\Delta(x) := \prod_{i < j} (x_i - x_j)$ is the Vandermonde determinant, and the constant $\kappa_\theta$ may again be determined by computing in two different ways a Gaussian integral over $\CM_{\theta,n}$ and
 making use of the Mehta--Dyson integral \cite{M-D}:
\be \int_{\R^n} dx\, |\Delta(x)|^{2\theta} e^{-\oh \sum_i x_i^2} = (2\pi)^{n/2} \prod_{j=1}^n \frac{\Gamma(1+j \theta)}{\Gamma(1+\theta)}\,,\ee
whence
\be\label{kappatheta} \kappa_\theta= \frac{(2\pi)^{\oh n(n-1)\theta}  \, n!}{\prod_{j=1}^n \frac{\Gamma(1+j \theta)}{\Gamma(1+\theta)}\,}\, .\ee

   If $\Delta(x) \neq 0$ we say that $x$ is {\it regular}; this means that $x$ is a diagonal (traceless) matrix with distinct eigenvalues.  For $x$ regular and $G \neq \SO(n)$ for $n$ even, the Jacobian of the diagonalization map is equal to the Riemannian volume of the orbit $\CO_x$ with respect to the metric induced by the inner product. (When $G = \SO(n)$, $n$ even, for regular elements there are two distinct orbits with the same spectrum, so that the Jacobian is equal to twice this volume.)
 
 As before, the Horn probability measure is supported on a convex polytope $\mathscr{H}_{\alpha \beta}$ in $\CC_+$, also called the {\it Horn polytope}, and is absolutely continuous with respect to the induced Lebesgue measure on $\mathscr{H}_{\alpha \beta}$. 
 We assume for the remainder of the paper that $\alpha$ and $\beta$ are regular, in which case $\dim \mathscr{H}_{\alpha \beta} = n-1$, so that the Horn probability measure has a density with respect to Lebesgue measure on $\CC_+$.

 For $\alpha, x \in \CC_+$, the orbital integral is again defined by
  \be\label{orb-int-selfadjoint}  \CH(\alpha, \ii x)  =\int_G dg\ e^{\ii \langle g \cdot \alpha, x \rangle}.
  \ee
{Following an analogous procedure to the derivation of (\ref{PDF-coadjoint}), we take the inverse Fourier transform of the characteristic function of the convolution of orbital measures, then multiply by the Jacobian of the diagonalization map to obtain} the PDF of the Horn probability measure:
  \be\label{PDF-selfadjoint}  
  p(\gamma | \alpha, \beta)=  
 \frac{\kappa_\theta^2\ \Delta(\gamma)^{2\theta}}{(2\pi)^{\dim \CM_{\theta,n}^0}n!}   \int_{\R^{n-1}} dx\ |\Delta(x)|^{2\theta} \CH(\alpha, \ii x) \CH(\beta, \ii x) \big(\CH(\gamma, \ii x)\big)^*, \qquad \gamma \in \CC_+,
  \ee
where $dx$ is the Lebesgue measure induced by identifying $\R^{n-1}$ with the hyperplane $\sum x_i = 0$ in $\R^n$.

 We expect to recover (\ref{PDF-coadjoint})  from (\ref{PDF-selfadjoint}) when $\theta = 1$. 
  In fact for $\SU(n)$ we have $\kappa_{\mathfrak{su}(n)} = \kappa_1$, $\Delta_{\mathfrak{su}(n)} = \Delta$, 
  $\dim  \CM_{\theta,n}^0=n^2-1 =\dim \mathfrak{su}(n)$, and $|W| = n!$, so indeed this is the case.

 By analogy with (\ref{eqn:J-def-coadjoint}), we define the volume function $\CJ(\alpha, \beta; \gamma)$ by
\bea\label{eqn:J-def-selfadjoint-a}
\CJ(\alpha, \beta; \gamma) &:=& \frac{(2\pi)^{\theta n (n-1)}}{\kappa_\theta^{2}\ \Delta_\gog(\rho)^3} \left( \frac{\Delta(\alpha) \Delta(\beta)}{\Delta(\gamma)} \right)^\theta \  p(\gamma | \alpha, \beta) \\
\label{eqn:J-def-selfadjoint-b}& =& \frac{\big ( \Delta(\alpha) \Delta(\beta) \Delta(\gamma) \big )^\theta}{(2 \pi)^{n-1} n!\, \Delta_\gog(\rho)^3} \int_{\R^{n-1}} dx\ |\Delta(x)|^{2\theta} \CH(\alpha, \ii x) \CH(\beta, \ii x) \big(\CH(\gamma, \ii x)\big)^*\,.
\eea
 Here $\Delta_\gog(\rho)$ indicates the same quantity as in the coadjoint case for the group $G$, so that this expression for $\CJ$  recovers (\ref{eqn:J-def-coadjoint}) for $\theta = 1$. 
 It is clear that $\CJ$ depends in both cases on the choice of $G$, but for the sake of brevity we choose {\it not} to append this information to the notation $\CJ$ and will instead specify the case and group under discussion whenever necessary.

Note that $\CJ$ is defined in (\ref{eqn:J-def-selfadjoint-a}) only for traceless $\alpha,\beta, \gamma$, but since 
in  (\ref{eqn:J-def-selfadjoint-b}) $\R^{n-1}$ is understood as the space of traceless $x$'s, 
$\ \CH(\alpha, \ii x)$ is invariant under translation of all $\alpha_i$ by the same constant $a$: 
\be\label{identite-bete} \CH(\alpha +aI, \ii x)  =  e^{\ii a  \sum x_i}\CH(\alpha, \ii x) =  \CH(\alpha, \ii x)\,,\ee
so that one may extend $\CJ$ to 
arbitrary $\alpha,\beta,\gamma$, even relaxing the conservation of traces.

\medskip

\subsection{Organization of the paper}
\begin{itemize}
\item In sect. \ref{self-adjoint-cases} we consider the self-adjoint case.  We first show in sect. \ref{sing-loci} that $\CJ$ is real-analytic away from a particular collection of hyperplanes, and that the equations defining these hyperplanes are the same in all three cases of real symmetric, complex Hermitian, and quaternionic self-dual matrices.  Next, in sect. \ref{nature-sing-SOn} we relate $\CJ$ to the Riemannian geometry of the orbits, considered as submanifolds of $V$, and we explain how this interpretation helps to understand the nature of the divergences that arise in $\CJ$ in the case of real symmetric matrices \cite{CZ2}. 
\item For the remainder of the paper after sect. \ref{self-adjoint-cases}, we restrict our attention to the coadjoint case.  In sect. \ref{sectJ-LR} we discuss the relationship between $\CJ$ and tensor product multiplicities, and we derive two different identities that express $\CJ$ in terms of tensor product multiplicities when the arguments are particular triples of highest weights.
\item In sect. \ref{sec:BZpol}, we relate $\CJ$ to the Euclidean volume of the {\it BZ polytope}, which provides a polyhedral model for tensor product multiplicities.  This point of view explains geometrically the relationship between $\CJ$ and tensor product multiplicities, and also provides insight into the nature of the non-analyticities of $\CJ$. It also leads to a proof that $\CJ$ does not vanish in the interior of the  Horn polytope.
\item Sect. \ref{sectB2} is a detailed case study of the above considerations for $B_2$, i.e. the case $\gog = \mathfrak{so}(5)$.
\end{itemize}


\section{The self-adjoint case}
\label{self-adjoint-cases}
In this section we take $G = \SO(n)$, $\SU(n)$ or $\mathrm{USp}(n)$, and we respectively fix $\theta = 1/2$, 1, or 2 and let $\CM_{\theta,n}^0$ be the set of traceless $n$-by-$n$ real symmetric, complex Hermitian or quaternionic self-dual matrices.  We study the action of $G$ on $\CM_{\theta,n}^0$ by conjugation.

In sect. \ref{sing-loci} we present an argument showing that non-analyticities of $\CJ$ lie along the same hyperplanes in all three cases.  In sect. \ref{nature-sing-SOn} we write $\CJ$ in terms of quantities related to the Riemannian geometry of the orbits, which can help to understand the origin of the divergences that appear in both $\CJ$ and the Horn PDF in the {real symmetric case.}


\subsection{Singular loci and nature of the non-analyticities}
   \label{sing-loci}  
     
   What follows is essentially an argument due to Mich\`ele Vergne \cite{MV}, based on a technique originally used to identify the singular loci of Duistermaat-Heckman densities in symplectic geometry.  It is well known \cite{DH} that for a Hamiltonian $G$-action on a symplectic manifold, the associated Duistermaat-Heckman measure on $\R^r$, $r = \mathrm{rank}(G)$, has a piecewise polynomial density with non-analyticities along certain hyperplanes.  In the coadjoint case the Horn probability measure is equal to a polynomial times the Duistermaat-Heckman measure for the diagonal $G$-action on $\CO_\alpha \times \CO_\beta$, so these symplectic methods can be used to identify the singular locus of the density.  In this section we show that an analogous technique works even in cases where the orbits do not carry a symplectic structure.
   
For $A, B \in \CM_{\theta,n}^0$ we can write $A=g_1 \alpha g_1^{-1}$, $B=g_2 \beta g_2^{-1}$ for some $g_1,g_2\in G$, where $\alpha$ and $\beta$ are diagonalizations of $A$ and $B$.
Given $\alpha$ and $\beta$  ordered and regular, and $g_1,g_2$ drawn independently at random from the Haar probability measure on $G$, we are interested in the distribution of  $\gamma = \diag(\gamma_i)$, where $\gamma_1 \ge \hdots \ge \gamma_n$ are the eigenvalues of $C=A+B$.

\begin{proposition}{\label{prop1} The distribution of $\gamma$ has a piecewise real-analytic density. 
Non-analyticities occur only when $\gamma$ lies on a hyperplane defined by an equation of the form
\be\label{sing-hyperpl}
\sum_{j\in K} \gamma_j= \sum_{j\in I} \alpha_j+ \sum_{j\in J} \beta_j\,,\ee
 where the subsets $I,\, J,\,K$ have the same cardinality: $|I|=|J|=| K|$. }
 \end{proposition}

\noindent \begin{proof} Following M. Vergne, we 
consider the map $\Phi : \ G\times G \to \CM_{\theta,n}^0$ that sends $(g_1,g_2) \mapsto C=A+B $. 
The pushforward by $\Phi$ of the product of the Haar measures is a
 $G$-invariant measure on the image, which has a density $\rho_{A,B}(C)$.  {The map $\Phi$ is proper and real-analytic, so that for $C_0$ any regular value of $\Phi$, a result of Shiga (see \cite{Shiga}, p.~133) guarantees the existence of a neighborhood $E \ni C_0$ such that $\Phi^{-1}(E) \cong \Phi^{-1}(C_0) \times E$ as real-analytic manifolds.  Let $(z, C)$ be local analytic coordinates on $\Phi^{-1}(C_0) \times E$.  In these coordinates the Haar measure has a real-analytic density $h(z, C)$, and for $C \in E$, $\rho_{A,B}$ can be written as $$\rho_{A,B}(C) = \int_{\Phi^{-1}(C_0)} h(z, C)\ dz,$$ which is clearly a real-analytic function on $E$.} It follows that $\rho_{A,B}$ is analytic in a neighborhood of every regular value of $\Phi$.  {As before, let $\CC_+ \subset \CM_{\theta,n}^0$ denote the cone of} traceless diagonal matrices $\mathrm{diag}(\gamma_1, \hdots, \gamma_n)$ with $\gamma_1 \ge \hdots \ge \gamma_n$. The restriction of $\rho_{A,B}$ to $\CC_+$ equals the volume function $\mathcal{J}$, up to a normalizing factor depending on $\alpha$ and $\beta$.  All non-analyticities of $\mathcal{J}$ must therefore occur at non-regular values of $\Phi$, i.e. $\gamma$ such that the differential $d\Phi$ fails to be surjective at some point in the preimage $\Phi^{-1}(\gamma)$.  The claim will now follow by identifying all non-regular values of $\Phi$.
 
At the point $(g_1, g_2)$, the differential is $$d_{(g_1, g_2)}\Phi (X, Y) = [X,A] + [Y,B], \quad X, Y \in \mathfrak{g}$$ where we have identified the tangent space $T_{(g_1, g_2)}(G \times G)$ with $\mathfrak{g} \oplus \mathfrak{g}$.  At a point where $d\Phi$ is non-surjective this operator has a non-trivial kernel, corresponding to a non-zero solution $Z \in \CM_{\theta,n}^0$ of
\be\label{easing} \forall X,Y \in \gog, \quad \tr(Z\cdot ([X,A]+ [Y,B]))=0\,. \ee
Thus we must determine for which $(g_1, g_2)$ such a $Z$ exists.
\\

(1) Using the invariance $\rho_{A,B} = \rho_{gAg^{-1}, gBg^{-1}}$, we can reduce to the case $g_1=I$ (at the price of redefining $X,Y,Z$).
Then taking $Y=0$,  the condition (\ref{easing}) reduces to $$\forall X\in \gog,\quad \tr( X\cdot [Z,{\diag(\alpha_i)}])=0,$$ so that we must have $[Z,{\diag(\alpha_i)}]=0$. Since the eigenvalues ${\alpha}_i$ are assumed distinct, this implies that $Z$ is diagonal. 

(2) Having taken $g_1 = I$, we rewrite $g_2 = g$.  For $X=0$, $Y$ arbitrary, the condition (\ref{easing}) reads $[Z,B]=0$. 
\\
-- If $Z$ is regular, this implies that $B$  is diagonal, and since $\beta$ is regular, this implies that 
$g$ acts as a permutation: $B=\diag(\beta_{w_i})$ for some $w \in S_n$. Thus the ordered eigenvalues $\gamma_i$ of $C$
are
\begin{equation}
\label{gammai} 
\gamma_i = \alpha_{w'_i}+\beta_{w_i}\qquad w,w'\in S_n, \quad i=1,\cdots,n \,, \end{equation}
which is a particular case of (\ref{sing-hyperpl}).\\
-- If $Z$ has repeated eigenvalues, with eigenvalue $z_\ell$ of multiplicity $r_\ell$ for $\ell=1,\cdots,s$, 
 one may only assert that $B$ is block-diagonal,
with $s$ blocks of size $r_1, r_2,\cdots r_s$. For each block of size 1, one is back to the situation described in
(\ref{gammai}). 
For each block $B_\ell$  of size $r_\ell>1$, the partial trace over 
the corresponding block in $\diag(\alpha)+B$  is a sum of $r_\ell$ eigenvalues $\gamma_j$, which we write $\sum^{(\ell)} \gamma_j =\sum^{(\ell)} \alpha_j +\tr B_\ell$. 
But since the trace of the block $B_\ell$ is just the sum of the  $\beta_j$ pertaining to 
that block, $\tr B_\ell=\sum^{(\ell)} \beta_j$, we have 
$\sum^{(\ell)} \gamma_j =\sum^{(\ell)} \alpha_j +\sum^{(\ell)} \beta_j$, 
which we can rewrite in the form (\ref{sing-hyperpl}). Such a linear relation on the $\gamma$'s defines a
hyperplane in $\R^{n-1}$,  since we have assumed that $A$, $B$ and $C$ were traceless. \end{proof}

\bigskip
\noindent {\bf Remarks}\\
1. 
The possible singularities identified in (\ref{sing-hyperpl}) include the hyperplanes that contain the facets of the Horn polytope (other than the walls of the Weyl chamber), where the Horn PDF and volume function vanish in a non-$C^\infty$ way.
The reader will recognize in (\ref{sing-hyperpl}) the form of Horn's (in)equalities.
\\
2. 
The singular hyperplanes in $\gamma$-space depend only on $\alpha$ and $\beta$ and 
not on $\theta$. This is in agreement with Fulton's argument \cite{Fu} that Horn's inequalities are the same for  all three cases considered in this section. This also justifies the empirical observation made previously
that the singular locus, for given $\alpha$ and $\beta$, is the same in all three cases \cite{Z1,CZ2}.
\\
3. 
Eq. (\ref{sing-hyperpl})  is only a {\it necessary} condition for a non-analyticity. It doesn't tell us
on which hyperplanes a non-analyticity does in fact occur. Also, it doesn't tell us that all the points of that hyperplane are singular.
An example is provided by the case $n=3$ where some singularities occur along half-lines in the $(\gamma_1,\gamma_2)$--plane
\cite{Z1, CZ2}. 
\\ 
4. The argument above doesn't tell us anything about the {\it nature} of the singularity. Indeed 
much stronger singularities appear in the real symmetric case  (where $\CJ$ can actually diverge) than in the complex Hermitian or quaternionic self-dual cases \cite{CZ2}.  For $\theta = 1$ or $2$, and for all of the coadjoint cases, we can use explicit formulae for the orbital integrals to write $\CJ$ as a sum of Fourier transforms.  A power-counting argument (differentiating under the integral sign and using the Riemann--Lebesgue lemma) then yields a lower bound on the number of continuous derivatives of $\CJ$.  In the complex Hermitian case with $n \ge 3$, one expects the function to be {\it at least} of differentiability class $C^{n-3}$ (see \cite{Z1}). For example, it is continuous but non-differentiable 
for SU(3), and at least once continuously differentiable for SU(4).  For SU(2), $\CJ$ is the indicator function of $\mathscr{H}_{\alpha, \beta}$ and is therefore discontinuous at the boundary. 
For a geometric interpretation of these singularities, see sect. \ref{nature-sing-SUn} below.


\subsection{Riemannian interpretation of $\CJ$, and singularities in the real symmetric case}
\label{nature-sing-SOn}
In this section we interpret $\CJ$ in terms of the Riemannian geometry of the orbits.  The Riemannian interpretation can help to understand the origin of the divergences that can appear in $\CJ$ {when $\theta = 1/2$. It was observed in \cite{Z1, CZ2} that for $\SO(2)$ and $\SO(3)$ acting on real symmetric matrices, $\CJ$ actually tends to infinity as it approaches certain singular hyperplanes.  It is unknown whether $\CJ$ diverges for $\theta = 1/2$ and $n > 3$.} We follow the notation of sect. \ref{sec:intro_selfadjoint}, and we assume as before that $\alpha$ and $\beta$ are regular and traceless.

The inner product $\langle A, B \rangle = \mathrm{tr}(AB)$ gives a $G$-invariant Riemannian metric on $\CM_{\theta,n}$, so that we obtain $G$-invariant induced metrics on the orbits $\CO_\alpha$ and $\CO_\beta$.   The associated Riemannian volume measures $\mu_\alpha$ and $\mu_\beta$ are also $G$-invariant, so they must respectively equal $\mu_\alpha(\CO_\alpha)$ and $\mu_\beta(\CO_\beta)$ times the unique invariant probability measure on each orbit.  If $\alpha$ and $\beta$ are both regular as we assume, then we have $\mu_\alpha(\CO_\alpha) = \kappa'_\theta \Delta(\alpha)^{2\theta}$, $\mu_\beta(\CO_\beta) = \kappa'_\theta \Delta(\beta)^{2\theta}$, where $\Delta$ is the Vandermonde determinant, and the constant $\kappa'_\theta$ equals $\kappa_\theta$ except when $\theta = 1/2$ and $n$ is even, in which case $\kappa'_\theta = \oh \kappa_\theta$.

Let $\tilde \Psi: \CO_\alpha \times \CO_\beta \to \CC_+$ be the map that sends $(A,B)$ to the diagonalization of $A+B$ with non-increasing entries down the diagonal.  Define the measure $\nu$ on $\CO_\alpha \times \CO_\beta$ as the product measure of the normalized $G$-invariant measures on each orbit. If we endow $\CO_\alpha \times \CO_\beta$ with the product metric of the induced Riemannian metrics, so that its volume measure is the product $\mu_\alpha \otimes \mu_\beta$, then we find $\nu = \kappa_\theta'^{-2} (\Delta(\alpha) \Delta(\beta))^{-2\theta}(\mu_\alpha \otimes \mu_\beta)$.  The Horn probability measure on $\CC_+$ is the pushforward $\tilde \Psi_* \nu = \kappa_\theta'^{-2} (\Delta(\alpha) \Delta(\beta))^{-2\theta} \tilde \Psi_*(\mu_\alpha \otimes \mu_\beta)$.

We can rewrite the measure in a simpler form by eliminating one of the orbits from the domain of $\tilde \Psi$.  Recalling that $\tilde \Psi$ is invariant under the diagonal $G$-action on $\CO_\alpha \times \CO_\beta$, it suffices to consider the case $A = \alpha$ and the ``reduced'' map $\Psi: \CO_\beta \to \CC_+$ that sends $B \in \CO_\beta$ to the diagonalization of $\alpha + B$ with non-increasing entries down the diagonal.  The Horn probability measure is then equal to $\kappa_\theta'^{-1} \Delta(\beta)^{-2\theta} \Psi_* \mu_\beta$.

If $\gamma_0$ is a regular value of $\Psi$, then for a sufficiently small coordinate neighborhood $E \ni \gamma_0$ all fibers of $\Psi$ over $E$ are diffeomorphic, and $\Psi^{-1}(E)$ is diffeomorphic to $\Psi^{-1}(\gamma_0) \times E$.  Let $z_1, \hdots, z_m$ be local coordinates on the fiber $\Psi^{-1}(\gamma_0)$, where $m = \dim \CM_{\theta,n} - n +1$.  Then $(z, \gamma)$ are local coordinates on $\Psi^{-1}(E)$, so that for $\gamma$ sufficiently close to $\gamma_0$ we can write the Horn PDF as the fiber integral
\be \label{eqn:pdf-fiber-int}
p(\gamma | \alpha, \beta) = \kappa_\theta'^{-1} \Delta(\beta)^{-2\theta} \int_{\Psi^{-1}(\gamma_0)} \sqrt{g_\beta(z, \gamma)}\ dz,
\ee
where $g_\beta$ is the determinant of the induced metric on $\CO_\beta$ in our chosen coordinates.  Accordingly, we have
\be \label{eqn:J-fiber-int}
\CJ(\alpha, \beta ; \gamma) = \frac{(2\pi)^{\theta n(n-1)}}{\kappa'_\theta \kappa_\theta^{2}\ \Delta_\gog(\rho)^3} \left( \frac{\Delta(\alpha)}{\Delta(\beta) \Delta(\gamma)} \right)^\theta \int_{\Psi^{-1}(\gamma_0)} \sqrt{g_\beta(z, \gamma)}\ dz
\ee
for $\gamma$ sufficiently close to any regular value $\gamma_0$ of $\Psi$.

Equation (\ref{eqn:J-fiber-int}) gives the desired Riemannian interpretation of $\CJ$ and provides some geometric insight into the origin of the volume function's singularities.  In particular, this point of view helps to explain why $\CJ$ can {actually diverge in the real symmetric case.}  The integral appearing in (\ref{eqn:pdf-fiber-int}) and (\ref{eqn:J-fiber-int}) looks almost like the induced volume of the (compact) fiber $\Psi^{-1}(y)$, so it may be surprising at first that for $\theta = 1/2$, $\CJ$ can tend to infinity on the interior of $\mathscr{H}_{\alpha \beta}$.  However, this integral is not the volume of the fiber.  If $g_\beta^\top$ and $g_\beta^\perp$ are the determinants, respectively, of the restriction of the induced metric on $\CO_\beta$ to the tangent bundle and normal bundle of $\Psi^{-1}(\gamma)$, then we have
$$\mathrm{Vol}(\Psi^{-1}(\gamma)) = \int_{\Psi^{-1}(\gamma_0)} \sqrt{g^\top_\beta(z, \gamma) }\ dz,$$
whereas
$$\int_{\Psi^{-1}(\gamma_0)} \sqrt{g_\beta(z, \gamma)}\ dz = \int_{\Psi^{-1}(\gamma_0)} \sqrt{g^\top_\beta(z, \gamma)\, g^\perp_\beta(z, \gamma) }\ dz.$$
The factor $g_\beta^\perp$ can blow up as $\gamma$ approaches a non-regular value of $\Psi$. {This merely reflects a singularity of the idiosyncratic choice of coordinates $(z, \gamma)$, but it will cause $\CJ$ to diverge if $g_\beta^\top$ doesn't diminish sufficiently to compensate.}

\bigskip
 To illustrate this idea, we consider the simple example of $\SO(2)$ acting on 2-by-2 real symmetric matrices,  relaxing temporarily our assumption that $\alpha$ and $\beta$ are traceless.  In this case, the orbits of regular elements are circles embedded in $\CM_{\frac{1}{2},2} \cong \R^3$.  If we parametrize $\SO(2)$ as rotation matrices $$R(\phi) = \begin{bmatrix} \cos \phi & -\sin \phi \\ \sin \phi & \cos \phi \end{bmatrix}, \quad 0 \le \phi < 2\pi,$$ then the orbit of $\beta = \mathrm{diag}(\beta_1, \beta_2)$ is
$$R(\phi) \beta R(\phi)^T = \begin{bmatrix} \beta_1 \cos^2 \phi + \beta_2 \sin^2 \phi & (\beta_1 - \beta_2) \cos \phi \sin \phi \\ (\beta_1 - \beta_2) \cos \phi \sin \phi & \beta_1 \sin^2 \phi + \beta_2 \cos^2 \phi \end{bmatrix}, \quad 0 \le \phi < 2\pi.$$
We have $\CC_+ = \{ \diag(\gamma_1, \gamma_2) \ | \ \gamma_1 \ge \gamma_2 \ \mathrm{and} \ \gamma_1 + \gamma_2 = \alpha_1 + \alpha_2 + \beta_1 + \beta_2 \},$ so that $\CC_+ \cong [0, \infty)$ is parametrized by the single coordinate $\gamma_{12} := \gamma_1 - \gamma_2$.  The map $\Psi$ sends $\phi \in [0, 2\pi)$ to $\gamma_1 - \gamma_2$ where $\gamma_1, \gamma_2$ are the eigenvalues of $\alpha + R(\phi) \beta R(\phi)^T$, and we have $$\Psi(\phi) = \sqrt{\alpha_{12}^2 + \beta_{12}^2 + 2\alpha_{12} \beta_{12} \cos(2\phi)},$$ where $\alpha_{12} := \alpha_1 - \alpha_2,$ $\beta_{12} := \beta_1 - \beta_2$.  The image of $\Psi$ is the interval $[ |\alpha_{12} - \beta_{12}|, \alpha_{12} + \beta_{12} ]$.  We assume that $\alpha_1 \neq \alpha_2$ and $\beta_1 \neq \beta_2$, as otherwise this image is just a single point.  An explicit computation yields
$$\CJ(\alpha, \beta; \gamma) = \begin{cases}
\frac{2}{\pi^2} \sqrt{\frac{\alpha_{12} \beta_{12} \gamma_{12}}{((\alpha_{12} + \beta_{12})^2 - \gamma_{12}^2)(\gamma_{12}^2 - (\alpha_{12} - \beta_{12})^2)}}, & \gamma_{12} \in [ |\alpha_{12} - \beta_{12}|, \alpha_{12} + \beta_{12} ], \\ 0, & \textrm{otherwise.} \end{cases}$$
Figure \ref{fig:J1-plot} shows the plot of $\CJ$ as a function of $\gamma_{12}$ for $\alpha_{12} = 1$, $\beta_{12} = 2$.

 \begin{figure}[htb]
  \centering
       \includegraphics[width=20pc]{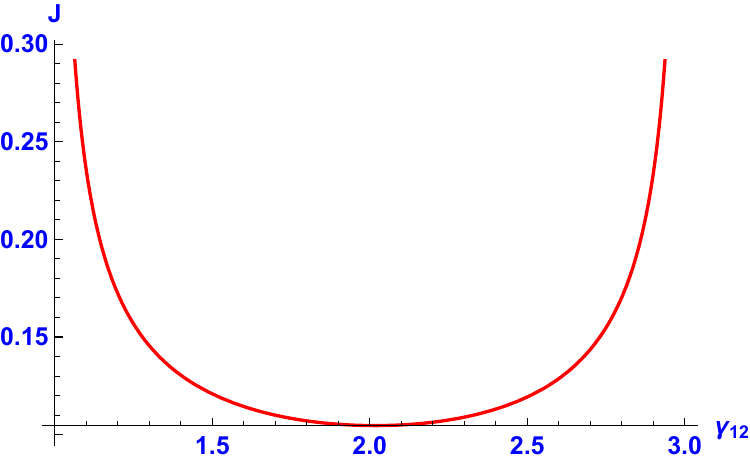}\caption{\label{J1-fig} The $\CJ$ function for $\alpha_{12}=1,\ \beta_{12}=2$.   }
        \label{fig:J1-plot}
       \end{figure}

$\CJ$ diverges at the endpoints $|\alpha_{12} - \beta_{12}| = \Psi(0)$ and $\alpha_{12} + \beta_{12} = \Psi(\pi)$, which are the non-regular values of $\Psi$.  The fiber of $\Psi$ over a regular value consists of two points, one in each of the open sets of $\CO_\beta$ parametrized by $0 < \phi < \pi$ and $\pi < \phi < 2\pi$.  In this case we may think of $\Psi$ as giving a coordinate chart on either of these two open sets; the density of the Riemannian volume form in these coordinates is just $\sqrt{g_\beta^\perp}$, which in this case is a function only of $\gamma_{12}$, for $\alpha_{12}, \beta_{12}$ fixed.  Since the fibers are 0-dimensional, the induced Riemannian volume on each fiber is just the counting measure, so that $\mathrm{Vol}(\Psi^{-1}(\gamma_{12})) = 2$ when $\gamma_{12}$ is a regular value.  Inverting the prefactor before the integral in (\ref{eqn:J-fiber-int}) and dividing by 2 to account for the volume of the fiber, we find that on either submanifold $\phi \in (0, \pi)$ or $\phi \in (\pi, 2\pi)$, the density of the Riemannian volume on $\CO_\beta$ can be expressed in terms of the local coordinate $\gamma_{12}$ as
$$\sqrt{g_\beta^\perp(\gamma_{12})} = \frac{\pi^3}{\sqrt{2}} \sqrt{\frac{\beta_{12} \gamma_{12}}{\alpha_{12}}} \CJ(\alpha_{12}, \beta_{12} ; \gamma_{12}).$$
The divergences of $\CJ$ at the endpoints indicate that this coordinate becomes singular as it approaches a non-regular value of $\Psi$.


\def\x{x}

\section{Relation $\CJ\leftrightarrow$ LR  in the coadjoint case}
\label{sectJ-LR}
For the remainder of this paper, we restrict our attention to the coadjoint case.
Let $G$ be a compact, connected, semisimple Lie group, $\gog$ its Lie algebra.  Here and in sect. \ref{sectVol} below, we will always make the further assumption that $\gog$ contains no simple summands isomorphic to $\mathfrak{su}(2)$; this assumption can be removed, but requires some additional care due to the discontinuity of $\CJ$ at the boundary of the Horn polytope in the $\mathfrak{su}(2)$ case (see \cite{CZ1}, sect. 4.1.1).

For $\x\in\mathfrak{t}$, a Cartan subalgebra, define\footnote{In the following, we use boldface $\Ga$ to denote roots, not
to confuse them with eigenvalues $\alpha_i$.}
\be\label{Deltas} \Delta_\gog(\x):= \prod_{\Ga >0} \langle \Ga,\x\rangle \qquad \hDelta_\gog(e^{\ii \x}):=  \prod_{\Ga >0} 
\Big(e^{\frac{\ii}{2}\langle \Ga,\x\rangle}- e^{-\frac{\ii}{2}\langle \Ga,\x\rangle  }\Big)\,.\ee
Then the Weyl--Kirillov formula for the character $\chi_\lambda$ of an  irreducible representation (irrep) $V_\lambda$ of 
highest weight (h.w.) $\lambda$ reads
\be\label{We-Ki} \frac{\chi_\lambda(e^{\ii \x})}{\dim V_\lambda}=\frac{\Delta_\gog(\ii \x)}{\hDelta_\gog(e^{\ii \x})} \CH(\lambda+\rho, \ii \x)\ee
with $\CH$ the orbital (Harish-Chandra) integral
\be\label{HC}  \CH(\lambda, \ii \x)=\int_G dg\, e^{\ii \langle \lambda, \mathrm{Ad}(g) \x\rangle}\,.\ee
Let $Q$ be the root lattice of $\gog$, $P$ the weight lattice, and $P^\vee$ the coweight lattice.  If $(\lambda,\mu,\nu)$ is a {\it compatible} triple of h.w., i.e. $\lambda+\mu-\nu \in Q$, and if we denote by primes 
$\lambda',\mu',\nu'$  the shift of $\lambda,\mu,\nu$ by the Weyl vector $\rho$, i.e. $\lambda'=\lambda+\rho$, etc., we have 
\bea\nonumber \CJ(\lambda',\mu';\nu') &=& \frac{ \dim V_\lambda \dim V_\mu \dim V_\nu }{(2\pi)^r |W|}
\int_{\mathfrak{t}} d^r \x\, |\Delta_\gog(\x)|^2\,  \CH(\lambda', \ii \x)\CH(\mu', \ii \x)(\CH(\nu', \ii \x))^*\\
\nonumber          &=& \frac{1}{(2\pi)^r |W|} \int_{\mathfrak{t}} d^r\x \,  |\hDelta_\gog(e^{\ii \x})|^2\, \frac{\hDelta_\gog(e^{\ii \x})}{\Delta_\gog({\ii} \x)} \, \chi_\lambda(e^{\ii \x}) \chi_\mu(e^{\ii \x})       (\chi_\nu(e^{\ii \x}) )^*\\
\nonumber  &=& \int_{\mathfrak{t} / 2 \pi P^\vee} d^r\x \, \frac{ |\hDelta_\gog(e^{\ii \x})|^2}{(2\pi)^r |W|}\sum_{\delta\in 2 \pi P^\vee}   \frac{\hDelta_\gog(e^{\ii (\x+\delta)})}{\Delta_\gog({\ii}(\x+\delta))} \, \chi_\lambda(e^{\ii { (\x+\delta)}}) \chi_\mu(e^{\ii {(\x+\delta)}})  (\chi_\nu(e^{\ii {(\x+\delta)}}) )^*. \eea
For $\delta \in 2 \pi P^\vee$, a direct computation yields the identity $\chi_\lambda(e^{\ii(x + \delta)}) = e^{\ii \langle \lambda, \delta \rangle} \chi_\lambda(e^{\ii x})$, and the compatibility of the triple $(\lambda,\mu,\nu)$ implies that $\exp \ii \langle (\lambda+\mu-\nu),\delta\rangle =1$.  Moreover, $\mathfrak{t} / 2 \pi P^\vee \cong \T / Z$, where $\T \subset G$ is the maximal torus with Lie algebra $\mathfrak{t}$, and $Z$ is the center of $G$.  Using these facts, we can rewrite the final line above as an integral over $\T$:
\bea \nonumber \CJ(\lambda',\mu';\nu') &=& \int_{\mathfrak{t} / 2 \pi P^\vee}   d^r\x \, \frac{ |\hDelta_\gog(e^{\ii \x})|^2}{(2\pi)^r |W|}     \left( \sum_{\delta\in 2 \pi P^\vee}  e^{\ii \langle \rho, \delta\rangle}   \frac{\hDelta_\gog({e^{\ii \x}})}{\Delta_\gog({\ii}(\x+\delta))}  \right)\, \chi_\lambda(e^{\ii \x}) \chi_\mu(e^{\ii  \x})  (\chi_\nu(e^{\ii  \x}) )^* \\
  \label{lasteq}&=&{ \int_\T       dT \left( \sum_{\delta\in 2 \pi P^\vee}  e^{\ii \langle \rho, \delta\rangle}   \frac{\hDelta_\gog({e^{\ii \x}})}{\Delta_\gog({\ii}(\x+\delta))}  \right)\, \chi_\lambda(e^{\ii \x}) \chi_\mu(e^{\ii  \x})  (\chi_\nu(e^{\ii  \x}) )^*}
           \eea
where$$dT=  \frac{\ |\hDelta_\gog(e^{\ii \x})|^2}{(2\pi)^r |W| |Z|} \, d^r\x.$$
By a similar calculation, we find that for $(\lambda,\mu,\nu)$ compatible,
\be\label{equnsh} \CJ(\lambda,\mu;\nu)= {\ \int_\T       dT \left( \sum_{\delta\in 2 \pi P^\vee} 
  \frac{\hDelta_\gog({e^{\ii \x}})}{\Delta_\gog({\ii}(\x+\delta))}  \right)\, \chi_{\lambda-\rho}(e^{\ii \x}) \chi_{\mu-\rho}(e^{\ii  \x})  (\chi_{\nu-\rho}(e^{\ii  \x}) )^*}\,. \ee

Now, as observed in \cite{CZ1} in the case of $G=\SU(n)$ and proved in full generality in \cite{ER}, the 
two sums over the coweight lattice $P^\vee$  that appear in (\ref{lasteq}, \ref{equnsh})
can each be expressed as a finite sum of characters over a set $K$, resp. $\hat K$, of h.w.
\be\label{CZER}  R := {\sum_{\delta\in 2\pi P^\vee}} e^{\ii \langle \rho, \delta\rangle} 
\frac{\hDelta_\gog({e^{\ii \x}})}{\Delta_\gog(\ii(\x+\delta))} = \sum_{\kappa\in K} \c_\kappa \chi_\kappa(T),\qquad 
\hat R := {\sum_{\delta\in 2\pi P^\vee}}\frac{\hDelta_\gog({e^{\ii \x}})}{\Delta_\gog(\ii(\x+\delta))} = \sum_{\kappa\in \hat K} \hat \c_\kappa \chi_\kappa(T)\,,\ee
where $\c_\kappa$ and $\hat\c_\kappa$ are coefficients to be determined, see below.
 According to \cite{ER}, in order for the character $\chi_\kappa$ to occur in the sum $R$ defined in (\ref{CZER}),  
$\kappa$ should belong to the interior of the convex hull of the Weyl group orbit of $\rho$.
{More precisely,}  $K$ is  the set of dominant weights   occurring in the irrep of h.w. $\rho-\xi$, where 
\be\label{xi} \xi=\begin{cases}  \sum_{i=1}^r \Ga_i \quad \mathrm{(sum\ of\ simple\ roots)} & \mathrm{ if}\   \rho\in Q\\
\sum_{i=1}^r k_i \Ga_i &  \mathrm{ if}\  \rho\notin Q ,\ \mathrm{with}\ \   
k_i=\begin{cases} 1 & \mathrm{if} \ \langle \rho, \omega_i^\vee\rangle \in \Z\\
\oh & \mathrm{if} \ \langle \rho, \omega_i^\vee\rangle \notin \Z
  \end{cases} 
\end{cases} \ee
 where the $\omega_i,\ i=1,\cdots, r$, are the fundamental weights.
 Observe that $\rho-\xi$, and therefore all weights  $\kappa\in K$, must lie in the root lattice.
 Obviously $\hat R$ and $R$ differ only if  $\rho\not\in Q$, {in which case $\hat K$ consists of the dominant weights occurring in the irrep of h.w. $\rho-\hat\xi$,} where $\hat\xi$ 
is obtained by  swapping the two lines above:
\be\label{xi}  \hat\xi= \sum_{i=1}^r \hat k_i \Ga_i  
\qquad  \mathrm{with}\ \   
\hat k_i=\begin{cases} \oh & \mathrm{if} \ \langle \rho, \omega_i^\vee\rangle \in \Z\\
1 & \mathrm{if} \ \langle \rho, \omega_i^\vee\rangle \notin \Z.
\end{cases} \ee
{Note that the trivial weight 0 always occurs in $K$, but never occurs in $\hat K$ when $\rho \not \in Q$.}

\bigskip
\noindent
{\bf Examples}. For SU(3) we have $\rho-\xi=0$, and $R  = \hat R$, given by the r.h.s. of   (\ref{CZER}), equals 1. For SU(4) we have $\rho-\xi=\omega_1+\omega_3$. {One finds $K = \{ (0, 0, 0),\ (1,0,1) \}$, with} $\c_0= 9/24$ and $\c_{\rho-\xi}=1/24$,  (see \cite{CZ1}, eq (62a)). 
 {One finds also $\rho-\hat \xi=\omega_2$, $\hat K = \footnotesize{\{(0, 1, 0)\}}$, and $ \hat \c_\kappa= \{1/6\}$  (see \cite{CZ1}, eq (62b)).}

More generally, for the $A_r$ series, one has $\xi = \sum_{i \, \text{odd}} \Ga_i/2 + \sum_{i \, \text{even}} \Ga_i$ {if $r$ is odd, and $\xi = \sum_{i } \Ga_i$ if $r$ is even.}
{Moreover $\hat \xi = \xi$ if $r$ is even and $\hat \xi = \sum_{i } \Ga_i$ if $r$ is odd.}
Explicit results for $R$, \ie for $\chi_\kappa$ and $\c_\kappa$, in the cases SU(5) and SU(6),  are also given in \cite{CZ1}, see sections 4.2.2, 4.2.3, 4.2.4, {as well as the results for $\hat R$ in the case of  SU(6)}.

 For $\SO(5)$, $\rho=\oh (3\Ga_1+4\Ga_2),\ \xi=\oh \Ga_1+\Ga_2$, $\rho-\xi=\Ga_1+\Ga_2=\omega_1$, $K=\{(0,0), (1,0)\}$, and  $R=\inv{8}(3\chi_0 +\chi_{\omega_1})$.
{One also finds $\hat \xi = \Ga_1+\Ga_2$,  $\hat K = \footnotesize{\{(0,1)\}}$ and {$\hat \c_{(0,1)} = 1/4$}.}

For $\SO(7)$, we find $K = \footnotesize{\{(0, 0, 0), (1, 0, 0), (0, 1, 0), (2, 0, 0), (0, 0, 2), (1, 1, 0), (1, 0, 2)\}}$ and the corresponding list of coefficients $\c_\kappa$: 
$\inv{92160}\{7230, 3995, 1651, 85, 479, 29, 1\}$.
{One also finds $\hat \xi = \Ga_1+\Ga_2 + \Ga_3$, $\hat K =  \footnotesize{\{(0, 0, 1),(1, 0, 1), (0, 1, 1)\}}$ and $\hat \c_\kappa =\inv{2880} \{190, 26, 1\}$.}

\bigskip
Introducing (generalized) Littlewood--Richardson (LR) coefficients
\begin{align*}
C_{\lambda\,\mu}^\nu &:= \dim \mathrm{Hom}_\gog(V_\lambda\otimes V_\mu\ \to V_\nu), \\
C_{\lambda\,\mu\,  \kappa}^\nu &:= \dim \mathrm{Hom}_\gog(V_\lambda\otimes V_\mu \otimes V_\kappa \to V_\nu),
\end{align*}
and putting (\ref{CZER}) into (\ref{equnsh}), we finally obtain: 
\begin{proposition}{\label{prop2}
Let $(\lambda,\mu;\nu)$ be a compatible triple of h.w., i.e.~such that $\lambda+\mu-\nu\in Q$. 
Let $(\lambda',\mu';\nu')$ be the corresponding triple shifted by the Weyl vector $\rho$: $\lambda'=\lambda+\rho$, etc.
Then we have the two
{\bf $\CJ$-LR  relations}:
\bea \label{In-LR} \CJ(\lambda',\mu';\nu') &=&\sum_{\kappa\in K, {\tau}} \c_\kappa  C_{\lambda\,\mu}^{{\tau}}  C_{{\tau}\,\kappa}^{{\nu}}
   =\sum_{\kappa\in K} \c_\kappa C_{\lambda\,\mu\,\kappa}^{\,\nu} \\
 \label{In-LRp}   \CJ(\lambda,\mu;\nu) &=&  \sum_{\kappa\in \hat K, {\tau}} \hat\c_\kappa  C_{(\lambda-\rho)\,(\mu-\rho)}^{{\tau}}  C_{{\tau}\,\kappa}^{{\nu-\rho}}
   =\sum_{\kappa\in \hat K} \hat\c_\kappa C_{(\lambda-\rho)\,(\mu-\rho)\,\kappa}^{\,\nu-\rho} 
\,. \eea}
\end{proposition}

\noindent {\bf Remarks.} \\ 1. 
The previous derivation generalizes and simplifies substantially the discussion given in \cite{CZ1} for  the case of $G=\SU(n)$.
 \\
 2. The coefficients $\c_\kappa,\, \hat\c_\kappa$ may be determined either by 
 a direct calculation of the sums in (\ref{CZER}), (as it was done in \cite{CZ1}), or by a geometric argument \cite{ER}, 
or by noticing that at the following special points, 
(\ref{In-LR},\ref{In-LRp}) reduce to:\footnote{The r.h.s. of (\ref{kissinger}), interpreted as a volume as we shall see in the next section, can be read, for instance, from the (stretched) LR polynomial defined by the triples that appear as arguments of $\CJ$, or, for low rank, computed from explicit expressions such as those in \cite{Z1} or below in sect. \ref{CJ2-SO5}.}

\be \label{kissinger} \c_\kappa= \CJ(\rho,\rho,\kappa+\rho) {,\ \kappa\in K\qquad\mathrm{and} \qquad \hat\c_\kappa=\CJ(\rho,\rho,\kappa+\rho) ,\ \kappa\in \hat K }  \,. \ee
3. Taking the limit $x\to 0$ in (\ref{CZER}), we see\footnote{
One may use  this relation
and the dimensions $\text{dim}\, (V_\kappa)$, $\kappa\in K$, to check the coefficients $\c_\kappa$. 
In the above examples of $B_2$ and $B_3$, the  dimensions $\text{dim}\, (V_\kappa)$ for the representations with $\kappa \in K$ are respectively $\{1, 5\}$ and $\{1, 7, 21, 27, 35, 105, 189\}$.}
that 
 \be\label{sumkappa} \sum_{\kappa\in  K} \c_\kappa \dim V_\kappa=1{\qquad\mathrm{and} \qquad  \sum_{\kappa\in \hat K} \hat\c_\kappa \dim V_\kappa=1} \,.\ee
 4. For $\nu$ deep enough in the dominant Weyl chamber, so that all the weights $\nu-\rho-k$ are dominant 
when $k$ runs over the set $\mathrm{wt}(\kappa)$ of weights of each $V_\kappa$, the r.h.s. of (\ref{In-LR}-\ref{In-LRp}) may be written explicitly
\be\label{deep-enough}
\CJ(\lambda,\mu;\nu)=\sum_{\kappa\in \hat K} \hat\c_\kappa  \sum_{k\in \mathrm{wt}(\kappa)}\,  \mult_\kappa(k)\,  C_{(\lambda-\rho)\,(\mu-\rho) }^{{\nu-\rho-k}} \,, 
 \ee
  and a similar formula for $\CJ(\lambda',\mu';\nu')$.  See examples  for $A_r$ in \cite{CZ1}  and for  $B_2$ in sect. \ref{many-ways} below. \\
  5. Multiplying (\ref{In-LR}) by $\dim V_\nu = \Delta_\gog(\nu')/\Delta_\gog(\rho)$, summing over $\nu$, and using (\ref{sumkappa}) along with the identity
  $\sum_\nu C_{\lambda \mu}^\nu \dim V_\nu = \dim V_\lambda \dim V_\mu,$
  one finds:
  $$\sum_\nu \CJ(\lambda', \mu' ; \nu') \frac{\Delta_\gog(\nu') \Delta_\gog(\rho)}{\Delta_\gog(\lambda') \Delta_\gog(\mu')} = \sum_\nu p(\nu' | \lambda', \mu') = 1,$$
so that the PDF (\ref{PDF-coadjoint}) also satisfies a discrete normalization condition.


\section{ Coadjoint case:  $\CJ$ as the volume of a polytope} 
\label{sectVol}
In this section we show that $\CJ(\alpha, \beta; \gamma)$ is equal to the (relative) volume of a certain convex polytope, the {\it Berenstein-Zelevinsky (BZ) polytope} $H_{\alpha \beta}^\gamma$, and we explore some consequences of this fact.  The primary importance of the BZ polytope is that the tensor product multiplicity $C_{\lambda \mu}^\nu$ is equal to the number of integer points in $H_{\lambda \mu}^\nu$.  This fact provides another perspective on the link between $\CJ$ and tensor product multiplicities.  In sect. \ref{sec:BZpol} we recall the definition of the BZ polytope, and we show in sect. \ref{sect:J-BZvol} that $\CJ$ computes its volume. In sect. \ref{sec:interior-nonvanish} we use this geometric interpretation to show that $\CJ$ cannot vanish on the interior of the Horn polytope, and in sect. \ref{nature-sing-SUn} we discuss how the non-analyticities of $\CJ$ arise from changes in the geometry of $H_{\alpha \beta}^\gamma$ as $\gamma$ varies.


\subsection{The BZ polytope}
\label{sec:BZpol}
Following Berenstein and Zelevinsky  {(BZ)} \cite{BZ1,BZ2,BZ3}, one may determine the LR coefficient $C_{\lambda \mu}^\nu$ pertaining to a compact or complex semisimple Lie algebra $\gog$ of rank $r$ by counting the number of integer points of a certain convex polytope, the {\it BZ polytope} 
(in the $A_r$ case it is closely related to the {\it hive polytope}\footnote{Actually the BZ-polytope is the image of the hive polytope under an injective lattice-preserving linear map \cite{PakValejo}, so that one can identify them for the purpose of counting arguments; however the Euclidean volumes of the two polytopes differ by an $r$-dependent constant.}  of \cite{KT99}),
which we denote $H_{\lambda\mu}^\nu$.  We will show below that the volume function $\CJ(\lambda,\mu;\nu)$ { is proportional to} the Euclidean volume of this polytope.  
Intuitively, for a compatible triple of h.w. $(\lambda, \mu, \nu)$, if the polytope $H_{\lambda\mu}^\nu$ is very large then we expect that the number of its integer points should give a very good approximation of its volume. In practice there are some additional subtleties because we want to count the integer points in a space of higher dimension than $H_{\lambda\mu}^\nu$, however at a heuristic level this intuition illustrates geometrically why $\CJ(\lambda,\mu;\nu)$ 
can be considered as a semiclassical approximation of $C_{\lambda \mu}^\nu$.

 The BZ construction hinges on the result that $C_{\lambda \mu}^\nu$ equals the number of ways of decomposing the weight $\sigma = \lambda + \mu - \nu$ as a positive integer combination of positive roots, 
  such that the decomposition also satisfies some additional combinatorial constraints.  To construct the polytope $H_{\lambda\mu}^\nu$, one therefore starts by introducing real parameters $\t_1, \hdots, \t_{N_r}$, where $N_r$ denotes the number of {\it positive} roots $\Ga_1, \hdots, \Ga_{N_r}$ of $\gog$.  A decomposition of $\sigma$ as a positive linear combination of positive roots corresponds to a point in the polytope of  {\it $\mathfrak{g}$-partitions with weight $\sigma$},
\be\label{part-poly}\mathrm{Part}_\mathfrak{g}(\sigma) = \big \{ (\t_1, \hdots, \t_{N_r}) : \sum_{a=1}^{N_r} \t_a \Ga_a = \sigma \big \} \cap \R^{N_r}_+, \ee
where $\R^{N_r}_+$ is the positive orthant in $\R^{N_r}$.  Since $\sigma$ lies in an $r$-dimensional space, we have $\dim \mathrm{Part}_\mathfrak{g}(\sigma) = N_r - r =: d_r$.  Positive {\it integer} decompositions of $\sigma$ correspond to integer points of $\mathrm{Part}_\mathfrak{g}(\sigma)$.  Additional linear constraints must still be imposed on these integer decompositions, so that the BZ polytope $H_{\lambda \mu}^\nu$ is finally obtained by intersecting $\mathrm{Part}_\mathfrak{g}(\sigma)$ with some number of half-spaces.   Generically we have $\dim H_{\lambda \mu}^\nu = d_r$, but for non-generic triples we may have $\dim H_{\lambda \mu}^\nu =: d \le d_r$.

One may alternatively introduce the quantity  
$f_r=N_r+2r$ that, in the case $A_r$, is the number of  ``independent fundamental intertwiners" (see  \cite[sect. 4]{CZ3}),
and then impose $3r$ conditions on the three weights $\lambda,\mu,\nu$ 
recovering $d_r = f_r-3r=N_r-r$.    
Finally, the $d_r$ independent parameters are again  subject to linear inequalities, 
 thus defining the convex polytope $H_{\lambda\mu}^\nu$.

 Note that we have defined $H_{\lambda \mu}^\nu$ as the solution set of a system of equations and inequalities that depend linearly on $\lambda,\ \mu$ and $\nu$.  We can therefore talk about the BZ polytope $H_{\alpha \beta}^\gamma$ associated to {\it any} triple of points $(\alpha, \beta, \gamma)$ in the dominant Weyl chamber, which may not be compatible highest weights or even rational points.
 
 In the case that $(\lambda, \mu, \nu)$ is indeed a compatible triple, $H_{\lambda\mu}^\nu$ is not in general integral, but it is always rational. Upon scaling by a positive integer $s$,
$P_{\lambda\mu}^{\nu}(s):=C_{s \lambda\,s\mu}^{s\nu}$ is a  {\it quasi-polynomial}  of the variable $s$, and is the Ehrhart quasi-polynomial
of the polytope $H_{\lambda\mu}^\nu$ \cite{Stanley}. 
It is sometimes called the {\it stretching 
quasi-polynomial} or {\it LR quasi-polynomial}. \\

{\bf Remarks}\\
1.   The definition of $P_{\lambda\mu}^{\nu}(s)$ makes sense whether or not  $(\lambda, \mu, \nu)$ is a compatible triple.
In many instances in the literature, like the construction of the BZ polytope and the discussion of saturation or of the properties of the stretching polynomial, compatibility of the triple is generally assumed. Since in the present paper we do not limit our consideration to compatible triples,\footnote{For instance, as noted in the previous section, for some algebras if a triple is compatible then the corresponding Weyl shifted triple is not.  This occurs for example in the case of $B_2$.} we will make clear the hypothesis of compatibility whenever it is necessary.\\
 2. 
When discussing such topics (stretching polynomials, saturation property, etc.) in terms of the representation theory of semisimple compact Lie groups rather than semisimple Lie algebras,  
one should be specific about the group under consideration since the conclusions will usually differ if one compares two Lie groups with the same Lie algebra but different fundamental groups. 
In the present paper, even though we consider the coadjoint representation of the orthogonal group $\SO(n)$, the tensor product multiplicities that arise in this setting for algebras of types $B$ or $D$ actually correspond to those of the simply-connected group Spin$(n)$.

\bigskip
\begin{center}
\begin{tabular}{|c|| c|c|c|c|c}
\hline
$\mathfrak{g}$ & $N_r$ & {$f_r$} & $d_r$ & $\delta_r$   \\
\hline\hline
$A_r$  & {$\oh r(r+1)$} & {$ \frac{1}{2} r (r+5)$} & $\oh r(r-1)$  & $(r + 1)^{r - 1}$ \\
\hline
$B_r$  &$ r^2$ & {$r (r+2) $}  & $r(r-1)$ & $(2 r - 1)^r$\\
\hline
$C_r$  &$r^2$ & {$r (r+2) $}   &$r(r-1)$  &  $2^{r - 2}  (r + 1)^r $\\
\hline
$D_r$&$ r(r-1)$ & {$r (1 + r)$}  & $r(r-2)$ & $2^{r - 2}  (r - 1)^r $\\
\hline
\end{tabular} 
\ {\begin{tabular}{|c|| c|c|c|c|c}
\hline
$\mathfrak{g}$ & $N_r$ & {$f_r$} & $d_r$ & $\delta_r$  \\
\hline\hline
$E_6$  & $36$ & $48$ & $30$ & $ 2^{12} 3^5$ \\
\hline
$E_7$   & $63$ & $77$  & $56$ & $2^ 6 3^{14} $ \\
\hline
$E_8$  & $120$ & $136$  & $112$ & $2^8 3^8 5^8 $ \\
\hline
$F_4$  & $24$ & $32$  & $20$ & $2^ 2 3^8 $ \\
\hline
$G_2$ & $6$ & $10$  & $4$ & $2^4 3$ \\ 
\hline
\end{tabular} }\\[8pt]
{Table 1. The numbers $N_r,\, f_r,\, d_r, \delta_r$ for the various simple algebras. The quantity $\delta_r$, discussed in the appendix, is a conjectured value of the squared covolume of the lattice $\Lambda$ defined in sect. \ref{sect:J-BZvol}. }
\end{center}


\subsection{$\CJ$ and the Euclidean volume of the BZ polytope }
\label{sect:J-BZvol}
 In this section we show that $\CJ(\alpha, \beta; \gamma)$ is proportional to the Euclidean $d_r$-volume of the BZ polytope $H_{\alpha \beta}^\gamma$ for an arbitrary compact or complex semisimple Lie algebra $\gog$ with no $\mathfrak{su}(2)$ summands. Specifically, we show that $\CJ$ equals the {\it relative $d_r$-volume} of the BZ polytope, defined as the Euclidean  $d_r$-volume divided by the covolume (volume of a fundamental domain) of the affine lattice $\Lambda := \mathrm{aff}(H_{\alpha \beta}^\gamma) \cap \Z^{N_r}$. Here $\mathrm{aff}(H_{\alpha \beta}^\gamma)$ indicates the affine span of $H_{\alpha \beta}^\gamma$, i.e. the minimal affine subspace of $\R^{N^r}$ containing $H_{\alpha \beta}^\gamma$, and $\Z^{N_r}$ is the integer lattice of the parameters $\t_a$ appearing in (\ref{part-poly}).  When $\dim H_{\alpha \beta}^\gamma = d_r$ we have $\mathrm{aff}(H_{\alpha \beta}^\gamma) = \mathrm{aff}(\mathrm{Part}_\gog(\sigma))$, and the covolume of $\Lambda$ is a constant $c_\gog$ that depends only on the root system of the algebra $\gog$, so that we have \be \rvol(H_{\alpha \beta}^\gamma) = \frac{1}{c_\gog} \mathrm{Vol}(H_{\alpha \beta}^\gamma), \ee where $\rvol$ is the relative $d_r$-volume and $\mathrm{Vol}$ is the Euclidean $d_r$-volume.   We discuss the covolumes $c_\gog$ in more detail in the appendix, where we explain the conjectured values for $\delta_r := c_\gog^2$ appearing above in Table 1.  Below when we refer to the ``volume'' of the BZ polytope, it will be understood that we mean the relative volume.  (Note that the relative volume is not the same thing as the {\it normalized volume} considered in \cite{CZ1}.)

   In the $A_r$ case (i.e.~$\SU(r+1)$ acting on traceless Hermitian matrices), the relationship between the volume function  $\CJ(\alpha, \beta; \gamma)$ and the Euclidean volume of $H_{\alpha \beta}^\gamma$ follows from the well-known fact that the stretching quasi-polynomials (for compatible triples) are genuine polynomials \cite{DW,Rass}.  In fact for $A_r$, if one considers the hive polytope rather than the BZ polytope, it turns out that $\CJ$ exactly computes the Euclidean volume, without a covolume factor.  Since this has been treated in several other places we shall not dwell on the matter, and instead refer the reader to the paper \cite{CZ1} for a more detailed discussion.  In the remainder of this section we treat the general case, namely:

\begin{proposition}
\label{prop:J-equals-BZvol}
Let $\gog$ be any compact semisimple Lie algebra without $\mathfrak{su}(2)$ summands. Then for any $\alpha, \beta, \gamma \in \CC_+$,
\be \CJ(\alpha, \beta; \gamma) = \rvol(H_{\alpha \beta}^\gamma). \ee
\end{proposition}

\begin{proof}

 We begin by assuming that we are working with a compatible triple $(\lambda, \mu, \nu)$ with $C_{\lambda \mu}^\nu \neq 0$, though later we will remove this assumption.
The stretching quasi-polynomial of $H_{\lambda \mu}^\nu$ can be written\footnote{For example, by inspection of the BZ inequalities for $B_2$ (see (\ref{BZ-ineq-B2}) below),
it is easy to see that in this case 
$H_{\lambda\mu}^\nu$ may have corners at integer or half-integer points, hence 
$P_{\lambda\mu}^{\nu}(s)$ is a quasi-polynomial of period 2. In fact it is well known that in the $B_r,\, C_r$ and $D_r$  cases, the period of  $P_{\lambda\mu}^{\nu}(s)$ is at most 2 \cite{DLMA}.} 
\begin{equation} \label{eqn:quasipoly-general} P_{\lambda\mu}^{\nu}(s) = \sum_{k=1}^{d} a_k(s) s^k,\end{equation} 
where $d = \dim H_{\lambda \mu}^\nu\le d_r$ and each $a_k$ is a rational-valued periodic function on $\N$.

\subsubsection*{Step 1:}
 Using the assumption that $C_{\lambda \mu}^\nu \neq 0$, it follows from results of McMullen \cite{McMullen} that the leading coefficient of the quasi-polynomial $P_{\lambda \mu}^\nu$ is constant.  At the end of this section we will sketch an intuitive argument for why this must be so, which does not require knowledge of  \cite{McMullen}.  The fact that $a_d$ is constant implies that it equals the $d$-volume of $H_{\lambda \mu}^\nu$, by the following simple observation.

Since $H_{\lambda \mu}^\nu$ is a rational polytope, we can choose $m \ge 1$ such that the $m$-fold dilation $m H_{\lambda \mu}^\nu$ is an integral polytope, whose Ehrhart polynomial is equal to 
$$P_{\lambda\mu}^{\nu}(ms) = a_{d} (ms)^{d} + \sum_{k=1}^{d - 1} a_k(m) (ms)^k.$$ By the standard result that the relative volume
of an integral polytope equals the leading coefficient of its Ehrhart polynomial, we then have $\rvol(mH_{\lambda \mu}^\nu) = a_{d} m^{d}$, and thus 
\be \label{eqn:vol-ad} \rvol(H_{\lambda \mu}^\nu) = a_{d}. \ee

 \subsubsection*{Step 2:} 
We may now use (\ref{eqn:vol-ad}) to identify the function  $\CJ$ with the $d_r$-volume of the BZ polytope. 
Upon dilation of $\lambda,\mu,\nu$ by a factor $s$, relation (\ref{In-LR}) gives
\be \label{eqn:J-scale-preapprox}
\CJ(s \lambda+\rho,s\mu+\rho;s\nu+\rho)   =
\sum_{\kappa\in K\atop \tau } \c_\kappa C_{s\lambda\,s\mu}^\tau C_{\tau\, \kappa}^{s\nu} 
\,.\ee
Since $\gog$ has no $\mathfrak{su}(2)$ summands, we can determine from the Riemann-Lebesgue lemma (see Remarks in sect. \ref{sing-loci}) that $\CJ$ is a continuous function of its three arguments.  For $s\gg1$, we use the continuity and homogeneity of $\CJ$ to approximate the l.h.s. of (\ref{eqn:J-scale-preapprox}) by $ \CJ(s \lambda,s\mu;s\nu)=s^{d_r}  \CJ( \lambda,\mu;\nu)$. On the r.h.s., we observe that for $s$ large, any weight $\tau=s\nu-k$,
where $k$ runs over the set $\mathrm{wt}(\kappa)$ of weights of $V_\kappa$, is dominant and contributes 
$\mult_\kappa(k)$  to  $ C_{\tau\, \kappa}^{s\nu}$,
\begin{equation} \label{eqn:J-sum-approx} \sum_{\kappa\in K\atop \tau } \c_\kappa C_{s\lambda\,s\mu}^\tau C_{\tau\, \kappa}^{s\nu} \  \buildrel{s\gg 1}\over{ = }
\sum_{\kappa\in K\atop k\in  \mathrm{wt}(\kappa)}  \c_\kappa \, \mult_\kappa(k)
\, C_{s\lambda\,s\mu}^{s\nu-k} \ \approx \ \sum_{\kappa\in K } \c_\kappa \dim V_\kappa\, C_{s\lambda\,s\mu}^{s\nu},\end{equation}
where we have approximated  $C_{s\lambda\,s\mu}^{s\nu-k}\approx  C_{s\lambda\,s\mu}^{s\nu}$.  This approximation is justified by the geometric observation that, since the equations and inequalities defining the BZ polytope depend linearly on $(\lambda, \mu, \nu)$, the difference between the number of integral points in $H_{s\lambda\,s\mu}^{s\nu}$ and in $H_{s\lambda\,s\mu}^{s\nu-k}$ must be lower order than $s^{d_r}$ for $s$ large.
Finally  using relation (\ref{sumkappa}), we obtain
$$\sum_{\kappa\in K} \c_\kappa   C_{s\lambda\,s\mu\, \kappa}^{s\nu} 
\approx 
C_{s\lambda\,s\mu}^{s\nu} = P_{\lambda\mu}^\nu(s) =  a_{d_r} 
s^{d_r}+\textrm{ lower-order\ terms,}$$
whence the identification
\be \CJ( \lambda,\mu;\nu) = a_{d_r}  =\rvol(H_{\lambda \mu}^\nu)\,. \ee 

Note that $\CJ(\lambda,\mu;\nu)$ may vanish for a compatible triple, but in this case by the above argument we have $a_{d_r} = 0$, so that these are exactly the cases when $d < d_r$ and the $d_r$-volume of $H_{\lambda \mu}^\nu$ vanishes. 

\subsubsection*{Step 3:}
 We now use a simple approximation technique to remove the assumption of compatibility and show that $\CJ( \alpha,\beta;\gamma) =\rvol(H_{\alpha \beta}^\gamma)$ for arbitrary points $\alpha,\beta,\gamma$ in the dominant chamber. It will be apparent from the proof of Proposition \ref{prop:J-positive-interior} below that if $\gamma \not \in \mathscr{H}_{\alpha \beta}$ then $H_{\alpha \beta}^\gamma$ is the empty set.  Accordingly we may assume $\gamma \in \mathscr{H}_{\alpha \beta}$. For any $\varepsilon > 0$ we can find a compatible triple $(\lambda, \mu, \nu)$ with $C_{\lambda \mu}^\nu \neq 0$ and $N \in \N$ such that $|\alpha - \lambda / N| + |\beta - \mu / N| + |\gamma - \nu / N| < \varepsilon$.  We have $$\CJ(\lambda / N, \mu / N ; \nu / N) = N^{-d_r} \CJ(\lambda, \mu; \nu) = N^{-d_r} \rvol(H_{\lambda \mu}^\nu) = \rvol(H_{\lambda/N \ \mu/N}^{\nu/N}),$$ and $| \rvol(H_{\lambda/N \ \mu/N}^{\nu/N}) - \rvol(H_{\alpha \beta}^{\gamma}) | = O(\varepsilon)$ for $\varepsilon$ small.  Since $\CJ$ is continuous, letting $\varepsilon \to 0$ we obtain $\CJ( \alpha,\beta;\gamma) =\rvol(H_{\alpha \beta}^\gamma)$.
 \end{proof}

\medskip
 We end this subsection by sketching a brief argument that the leading coefficient $a_d$ of $P_{\lambda \mu}^\nu$ must be constant when $(\lambda, \mu, \nu)$ is a compatible triple with $C_{\lambda \mu}^\nu \neq 0$.  For the sake of simplicity we will assume that $P_{\lambda\mu}^{\nu}$ has period 2, so that it has the form \begin{equation} \label{eqn:quasipoly-period2} P_{\lambda\mu}^{\nu}(s) = \sum_{k=1}^{d} (a^+_k + (-1)^s a^-_k)s^k \end{equation} where $a^+_k, a^-_k$ are rational coefficients. However, the discussion easily extends to an arbitrary period.

Since $C_{\lambda \mu}^\nu \neq 0$, $\mathrm{aff}(H_{\lambda \mu}^\nu)$ obviously contains at least one integer point.  Moreover $\mathrm{aff}(H_{\lambda \mu}^\nu)$ is a $d$-dimensional rational affine subspace, so the fact that it contains one integer point implies that it contains a rank $d$ affine sublattice $\Lambda \subset \Z^{N_r}$.  Choose a linear transformation $\psi: \R^{N_r} \to \R^{d}$ that maps $\Lambda$ bijectively to $\Z^{d}$.  For all $s = 1, 2, ...$, the number of integer points in $s H_{\lambda \mu}^\nu$ is equal to the number of integer points in $s \psi(H_{\lambda \mu}^\nu)$, so that these two polytopes have the same Ehrhart quasi-polynomial $P_{\lambda \mu}^\nu(s)$. 
 Thus we have reduced the problem to studying dilations of a {\it full} polytope (that is, a $d$-dimensional polytope in $\R^{d}$ rather than the higher-dimensional space $\R^{N_r}$). 

Now suppose for the sake of contradiction that the leading coefficient of $P_{\lambda \mu}^\nu(s)$ were not a constant, i.e.~$a_d^- \ne 0$ in (\ref{eqn:quasipoly-period2}), and compute the difference 
$P_{\lambda\mu}^\nu(s+1)-P_{\lambda\mu}^\nu(s) = 2a_d^- s^d +\ \textrm {(lower-order terms)}$.
This equals the difference between the numbers of integer points in the two polytopes $(s+1) \psi(H_{\lambda \mu}^\nu)$ and $s\psi(H_{\lambda \mu}^\nu)$. It is easy to see that this number is bounded by
a multiple of the $(d-1)$-dimensional surface area of the larger polytope and is therefore $O(s^{d-1})$, in contradiction with the previous expression.  This proves that the leading coefficient of $P_{\lambda \mu}^\nu(s)$ must be a constant.

 
\subsection{Non-vanishing of $\CJ$ on the interior of the Horn polytope}
\label{sec:interior-nonvanish}
Proposition \ref{prop:J-equals-BZvol} leads to a proof of the following proposition, which generalizes a result that was shown for the $A_r$ case in \cite{CZ1}.
\begin{proposition}
\label{prop:J-positive-interior} For all $\alpha, \beta \in \CC_+$ and $\gamma$ in the interior of $\mathscr{H}_{\alpha \beta}$ (in the topology of $\mathfrak{t}$), $\CJ(\alpha, \beta ; \gamma) > 0$.
\end{proposition}

\begin{proof}
We may assume that $\dim \mathscr{H}_{\alpha \beta} = r$; otherwise its interior is empty.  It follows from Proposition \ref{prop:J-equals-BZvol} that $\CJ(\alpha, \beta ; \gamma) > 0$ exactly when $\dim H_{\alpha \beta}^{\gamma} = d_r$.  There is at least one point $\gamma_\star$ in the interior such that $\dim H_{\alpha \beta}^{\gamma_\star} = d_r$, since $\CJ$ is locally a polynomial of degree $d_r > 0$ and cannot vanish everywhere. It thus suffices to show that $d(\gamma) := \dim H_{\alpha \beta}^\gamma$ is constant on the interior of $\mathscr{H}_{\alpha \beta}$.

The BZ polytope $H_{\alpha \beta}^\gamma$ is the simultaneous solution set of the equation 
\begin{equation} \label{eqn:BZ-eq} \sum_{a=1}^{N_r} \t_a \Ga_a = \alpha + \beta - \gamma \end{equation} 
and a collection of $m$ linear inequalities that can all be written in the form 
\begin{equation} \label{eqn: BZ-ineqs-generic} v_j(\t_1, \hdots, \t_{N_r}) \ge h_j(\alpha, \beta, \gamma), \quad j = 1, \hdots, m,\end{equation} 
where each $v_j$ is a linear functional on $\R^{N_r}$ and $h_j$ is a linear functional on $\R^{3r}$.  Moving all terms involving $\gamma$ to the left-hand side, we can rewrite (\ref{eqn:BZ-eq}) and (\ref{eqn: BZ-ineqs-generic}) as 
\begin{align}  \label{eqn:univ-BZ1}  & \sum_{a=1}^{N_r} \t_a \Ga_a + \gamma = \alpha + \beta, \\ \label{eqn:univ-BZ2} & v_j'(\t_1, \hdots, \t_{N_r}, \gamma) \ge h_j'(\alpha, \beta), \quad j = 1, \hdots, m, & \end{align} where each $v_j'$ is a linear functional on $\R^{N_r + r}$ and $h_j'$ is a linear functional on $\R^{2r}$.  Considering $\gamma$ as an independent variable, the simultaneous solution set of (\ref{eqn:univ-BZ1}) and (\ref{eqn:univ-BZ2}) is a polytope in $\R^{N_r + r}$, which we denote $U_{\alpha \beta}$.  The projection $\R^{N_r + r} \to \R^r$ onto the $\gamma$ coordinates maps the relative interior of $U_{\alpha \beta}$ onto the interior of $\mathscr{H}_{\alpha \beta}$.  Let $P$ be this map on the interiors.  Each fiber $P^{-1}(\gamma)$ is in turn the relative interior of $H_{\alpha \beta}^\gamma$, so that $\dim P^{-1}(\gamma) = \dim H_{\alpha \beta}^\gamma$.  Moreover the map $P$ is a global submersion, {so that} its fibers all have the same dimension (see e.g. \cite{Lee} ch. 7). This completes the proof.
\end{proof}


\subsection{Geometric origin and nature of the non-analyticities of $\CJ$}
\label{nature-sing-SUn}

We now explain how the non-analyticities of $\CJ$ can be understood in terms of the geometry of $H_{\alpha \beta}^\gamma.$  The facets of $H_{\alpha \beta}^\gamma$ are cut out by some number of hyperplanes in $\mathbb{R}^{N_r}$ corresponding to the BZ inequalities.  For fixed $\alpha$ and $\beta$, these hyperplanes undergo linear translations when $\gamma$ varies in the dominant chamber.  As $\gamma$ varies, an inequality may become redundant, meaning that the corresponding hyperplane no longer intersects the polytope, so that the polytope has one fewer facet;  or alternatively, a previously redundant inequality may become relevant, meaning that the corresponding hyperplane intersects the polytope, forming a new facet. 
Non-analyticities of the volume $\CJ$ as a function of $\gamma$ occur at such points, where one of the hyperplanes defined by the BZ inequalities hits the polytope $H_{\alpha \beta}^\gamma$, or conversely does not intersect it anymore.

Since we know that $\CJ$ is a piecewise polynomial function of $\gamma$, these non-analyticities take the form of a change of polynomial determination (\ie the local polynomial form of $\CJ$). 
For a point $\gamma$ at a distance $\varepsilon$ from a non-analyticity hyperplane (see Proposition 1), the change of determination is of
the form $\Delta \CJ= O(\varepsilon^m)$, which is the volume of the piece of the polytope chopped off by the incident hyperplane.
In the neighbourhood of that non-analyticity hyperplane, the function is thus of differentiability class $C^{m-1}$. 
  Obviously the integer $m$ is bounded from above by $d_r$, but also from below  due to known lower bounds on the number of continuous derivatives of $\CJ$ (see the Remarks in sect. \ref{sing-loci}). One finds $r-1\le m \le d_r=\oh r(r-1)$
  for the $A_r$ cases, and  $ 2(r-1) \le m\le d_r=r(r-1)$ for $B_r$. For example in the case of SU(4), an explicit (unpublished) calculation
  with $\alpha=\beta=(3,2,1,0)$
  has revealed non-analyticities of class $C^{1}$ and $C^2$.  See Fig. \ref{sing-vol-fig} for illustration,
where the right-most case is in fact prohibited by the above bound on $m$. One may convince oneself that this bound on $m$ guarantees that any non-analyticity of the volume of the polytope must involve  the appearance or disappearance of a facet, and not merely a rearrangement of lower-dimensional faces.

  \begin{figure}[htb]
  \centering
       \includegraphics[width=40pc]{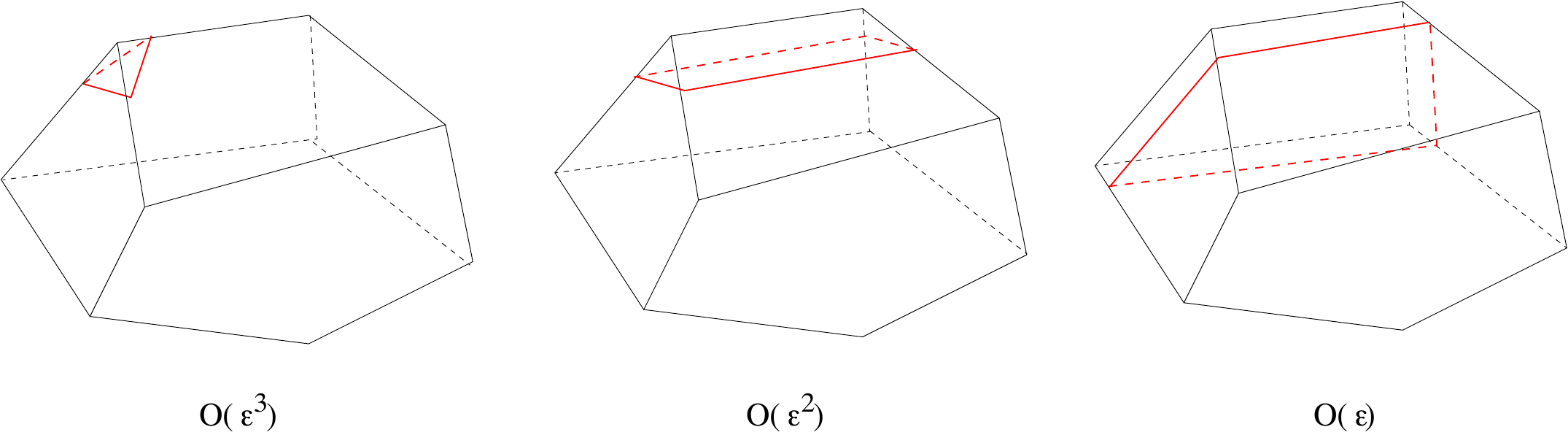}\caption{\label{sing-vol-fig} Artist's view of the intersection of 
       a $(d-1)$-hyperplane with the polytope  $H_{\alpha \beta}^\gamma$, here in $d=3$ dimensions, chopping off 
       a volume of order $O(\varepsilon^3)$, $O(\varepsilon^2)$ or $O(\varepsilon)$.
       We argue in the text that the right-most case cannot occur. }
       \end{figure}

 \section{The case of $B_2$}
 \label{sectB2}
 In this section, we illustrate the previous considerations in the case of $B_2=\mathfrak{so}(5)$.
 
 \def\smat#1{\mbox{\footnotesize{\mbox{$\begin{pmatrix}#1\end{pmatrix}$}}}}
 \subsection{The function $\CJ$ for $\SO(5)$}
 \label{CJ2-SO5}  
 Consider two {\it skew--}symmetric $5\times 5$ real matrices, in the block diagonal form $A=
\diag\bigg( {\footnotesize\begin{bmatrix}0&\alpha_i\\ -\alpha_i&0 \end{bmatrix}}_{i=1,2}, 0\bigg)$
and likewise for $B$.
 For $\SO(5)$ orbits of  these matrices,    the function $\CJ$ has been written in \cite{Z1} in the form 
\bea\nonumber
\CJ(\alpha,\beta;\gamma)&=&\frac{1}{64 \pi^2} \int \frac{ds\, dt }{s t (t^2-s^2)} \left[ \sin  s(\alpha_1+\alpha_2) \sin t(\alpha_1-\alpha_2)
-\sin  s(\alpha_1-\alpha_2) \sin t(\alpha_1+\alpha_2) \right]\\
&& \qquad \qquad\qquad \qquad \qquad \qquad \times \left[ \mathrm{same\ with}\ \beta\right]\left[ \mathrm{same\ with}\ \gamma\right]\,,\eea
which may then be written explicitly as a degree 2 piecewise polynomial. (Here the PDF, normalized on 
 the Horn polytope $\mathscr{H}_{\alpha\beta}$, is given by $\frac{3}{2} 
 \frac{|\Delta_O(\gamma)|}{|\Delta_O(\alpha)|\,|\Delta_O(\beta)|}\CJ$.)
Recall that the $B_2$ Weyl group $W=\mathrm{Dih}(4)=S_2\ltimes (\Z_2\times \Z_2) $  acts on the 2-vector $(\alpha_1,\alpha_2)$ 
by a change of sign of either component, 
or by swapping them. Denote the action of $(w, w', w'') \in W^3$ on the vector $\alpha+\beta-\gamma$ by
  $\sigma :=w(\alpha)+w'(\beta)-w''(\gamma)$,  and thus $\sigma_i=w(\alpha)_i+w'(\beta)_i-w''(\gamma)_i$, $i=1,2$. Let $\epsilon(w)$ denote the sign of a Weyl group element. Then
 \be\label{CJ2expl-alt} \CJ= \inv{2^8} \sum_{w,w',w''\in W} \epsilon(w) \epsilon(w') \epsilon(w'') \Big(4\sigma_1 |\sigma_1| -4 \sigma_2 |\sigma_2|   -2(\sigma_1-\sigma_2) |\sigma_1-\sigma_2|  \Big)\,  \sign(\sigma_1+\sigma_2)\,. \ee
 (Note that one may get rid of one of the three summations over the Weyl group, fixing one of the $w$'s to the identity and
multiplying the result by a factor $|W|=8$, which simplifies greatly the actual computation.)

It follows from this expression that $\CJ$  is of differentiability
class $C^{1}$ \footnote{There is an unfortunate misprint in sect. 5 of \cite{Z1}: the function $\CJ$ of $B_r=\mathfrak{so}(2r+1)$ is of class
$C^{2r-3}$, again as a consequence of the Riemann--Lebesgue lemma.}. 
Recall that one may always assume that $\alpha_1 > \alpha_2 > 0,\ \beta_1 > \beta_2> 0$ 
and $\gamma_1\ge \gamma_2\ge 0$. 
The support of $\CJ$ is determined by generalized  Horn inequalities of $B_2$ type. To write them down, 
we note that if $A$ is  a skew--symmetric  real matrix, then $\ii A$ is a complex Hermitian matrix. Thus the inequalities of
$B_2$ type 
 follow from the classical Horn's inequalities  \cite{Ho, Fu} of $A_4$ type (\ie for  $5\times 5$ complex Hermitian matrices), 
applied to matrices of the form $\mathrm{diag}(\alpha_1,\alpha_2,0,-\alpha_2,-\alpha_1)$, and likewise for $\beta$ and $\gamma$. One finds 
\bea\label{Horn-ineq-B2} \max( |\alpha_1-\beta_1|, |\alpha_2-\beta_2|) \le & \gamma_1 &\le \alpha_1+\beta_1\\
  \max(0,  \alpha_2-\beta_1, -\alpha_1+\beta_2)  \le & \gamma_2 &\le  \min( \alpha_1+\beta_2,\alpha_2+\beta_1)  \\
 |\alpha_1-\beta_1| +|\alpha_2-\beta_2|   \le & \gamma_1 +\gamma_2&\le \alpha_1+\alpha_2+\beta_1 +\beta_2 \\
   \max( 0, \alpha_1-\alpha_2-\beta_1-\beta_2, \beta_1-\beta_2-\alpha_1-\alpha_2)  \le & \gamma_1 -\gamma_2&\le \alpha_1+\beta_1-|\alpha_2-\beta_2| \,.
   \eea
 These inequalities, supplemented by $\gamma_1\ge \gamma_2\ge 0$, define the ($B_2$--generalized) Horn polygon, see Fig. \ref{twocases} for examples. 
     
The  singular lines of $\CJ$ may be determined by the same kind of argument as in sect. \ref{sing-loci}, or, following the same reasoning as above, as a special case of the singular lines of $A_4$ type.
Thus the {\it possible} lines of non-$C^2$ differentiability are those among 
 \bea  \nonumber
\gamma_1= \gamma_{1s}&\in&\{\alpha_1+\beta_2\,,\quad \alpha_2+\beta_1\,, \quad\alpha_2+\beta_2\,,\quad    |\alpha_1-\beta_2|\,, \quad |\alpha_2-\beta_1|\}\,,\quad   \\  \nonumber
  \gamma_2 = \gamma_{2s}&\in&\{\alpha_2+\beta_2\,,\quad   |\alpha_1-\beta_2|\,, \quad |\alpha_2-\beta_1|\,,\quad |\alpha_2-\beta_2|\,, \quad |\alpha_1-\beta_1|\}\ \\
  \label{sing-lines} 
  \gamma_1 + \gamma_2 = \gamma_{1+2\,s}&\in&\{ \alpha_1+\alpha_2+\beta_1-\beta_2 \,, \    | \alpha_1+\alpha_2-\beta_1+\beta_2| \,, \  \alpha_1-\alpha_2+\beta_1+\beta_2 \,, \\ 
  \nonumber && \qquad \qquad \qquad \qquad \qquad \qquad \   | -\alpha_1+\alpha_2+\beta_1+\beta_2|\,,\  \alpha_1-\alpha_2+\beta_1-\beta_2 \} \\
\nonumber \gamma_1-   \gamma_2 = \gamma_{1-2\,s}&\in&\{
  |- \alpha_1+\alpha_2+\beta_1+\beta_2|\,,\  | \alpha_1+\alpha_2-\beta_1+\beta_2| \,, \    \alpha_1-\alpha_2+\beta_1-\beta_2 \,, \
  \\ 
 \nonumber  && \qquad \qquad \qquad \qquad \qquad \qquad    | \alpha_1-\alpha_2-\beta_1+\beta_2| \,, \   | \alpha_1+\alpha_2-\beta_1-\beta_2| \}
\eea
that intersect the Horn polygon.\\

 { One may find an explicit piecewise polynomial representation of the volume function $\CJ$. At the price
  of swapping $\alpha$ and $\beta$, one may always assume that 
  $$ |\beta_1 -\alpha_2| \ge |\alpha_1-\beta_2|\,.$$
  The lines (\ref{sing-lines}) have 4-fold  intersections at four  vertices, denoted $I,J,K,L$, with coordinates
 $$ I=(\alpha_1 + \beta_2, |\alpha_2 - \beta_1|); \ J= (\alpha_2 + \beta_1,  |\alpha_1 - \beta_2|);\
K=(| \beta_1-\alpha_2 |, |\alpha_1 - \beta_2|); \  L= (\alpha_2 + \beta_2, |\alpha_1 - \beta_1|),$$
 some of which may be outside the Horn  polygon.

 \begin{figure}[bht]
  \centering
       \includegraphics[width=18pc]{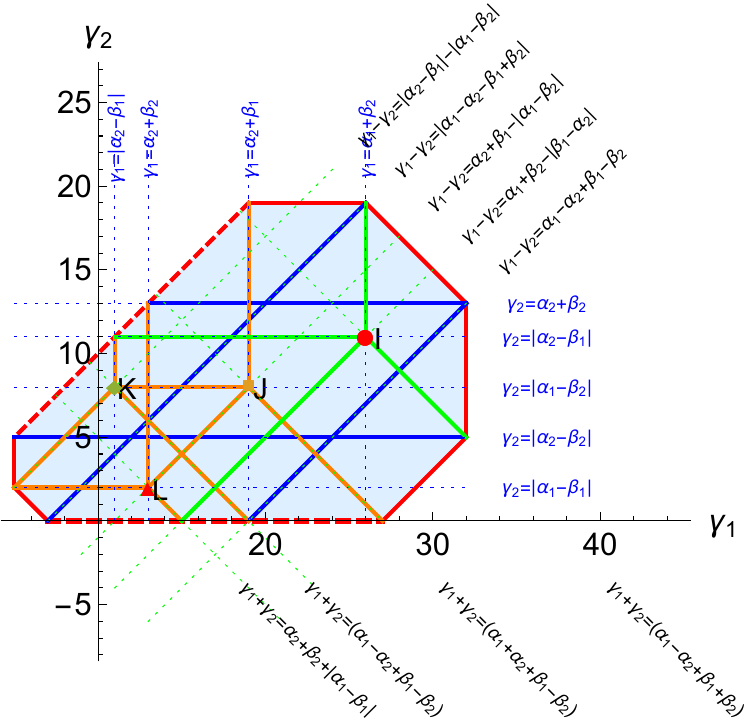} \includegraphics[width=18pc]{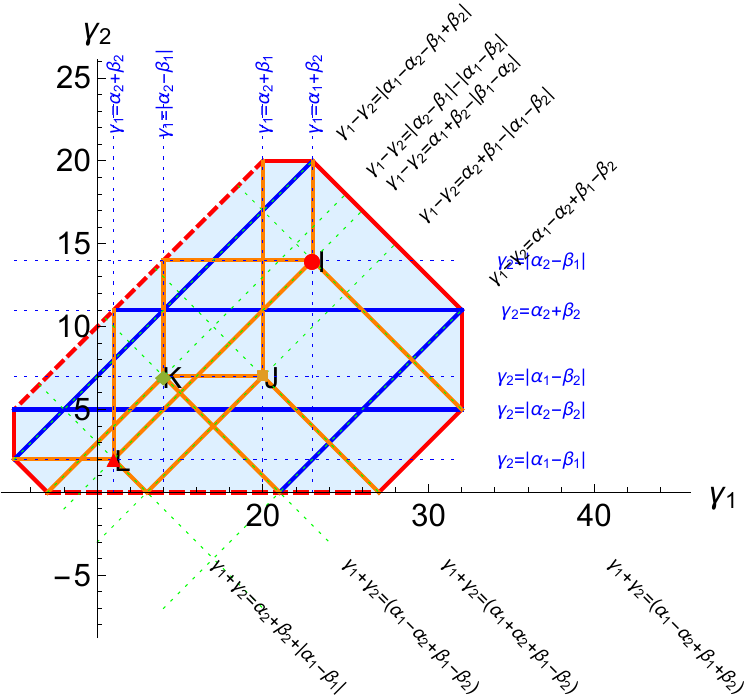}  
       \caption{\label{twocases}  Two examples of the Horn polytope and of the singular lines for
       $\alpha=(17,4),\, \beta=(15,9)$ and for $\alpha=({15,3}),\, \beta=({17,8})$.
       The function $\CJ$ has a quadratic change of determination across the solid lines, and a linear change 
       across the two dashed edges of the Horn polytope. }
       \end{figure}
  
Then two cases arise, depending on whether $\alpha_1\ge \beta_1$ or $\beta_1\ge \alpha_1$. 
In the latter case, and only then, $I$ and $L$ belong to the same diagonal $\gamma_1-\gamma_2= 
\mathrm{constant}$. A  fairly generic example of
each case is depicted in Fig. \ref{twocases}. (Some of the lines might merge or disappear from the figure for
other values of $\alpha$ or $\beta$.) 
The heavy solid lines shown in the figure are the loci of singularities, where  the polynomial determination changes. Note that these lines join at ``four-prong vertices", either inside the polygon (Fig. \ref{four-prong}.a), at some of
the points $I,J,K$ or $L$, or at its boundary (Fig. \ref{four-prong}.b).

Rather than giving a detailed polynomial expression in each sector arising from this decomposition, it is simpler to give an empirically observed set of rules that determine the {\it change} of polynomial determination.  These rules are as follows.

Starting from the exterior of the polygon, where $\CJ$ vanishes, enter the polygon through any of the solid red lines.
Crossing any of the solid lines along the direction of the arrow increases $\CJ$ by $\oh \Delta^2$, 
where $\Delta=(\gamma_i -\gamma_{is})$ for a vertical or horizontal line of equation $\gamma_i=\gamma_{is}$, and 
 $\Delta=\inv{\sqrt{2}}(\gamma_1\pm \gamma_2 -\gamma_{1\pm2\, s})$ for an oblique line of equation $\gamma_1\pm \gamma_2 =\gamma_{1\pm2\, s}$. 
 The arrows are shown  in Fig. \ref{four-prong}.
 When two lines merge, the differences $\oh \Delta^2$ add up.
 
 Now the reader will verify that this prescription is consistent: \\
-- Following a closed loop around a four-prong vertex, see Fig. \ref{four-prong}, we return to our initial expression for $\CJ$ (ensuring that the rules actually give a well-defined function), thanks to the trivial identity
$$ x^2+y^2 = \Big(\frac{x+y}{\sqrt{2}}\Big)^2+  \Big(\frac{x-y}{\sqrt{2}}\Big)^2 $$
where $x=\gamma_1-\gamma_{1s}$ and $y=\gamma_2-\gamma_{2s}$. \\
--  By considering a path that crosses the polygon, this implies in particular that} $\CJ$ returns to 0  outside the polygon.

Also note that along the edges $\gamma_1-\gamma_2=0$ or $\gamma_2=0$ of the polygon, which are {\it not} 
dictated by Horn's inequalities but rather reflect our choice to work in the dominant Weyl chamber, $\CJ$ may vanish only linearly. This is why those
lines have been  depicted by dashed lines in Fig. \ref{twocases}, and is also why no prescription is given for the crossing of those dashed lines. 

 At this stage, this piecewise polynomial construction of $\CJ$ is just a conjecture that has been checked on many examples.
In principle, establishing it should follow from a careful examination of eq. (\ref{CJ2expl-alt}) and of its possible changes of determination across the singular lines given in (\ref{sing-lines}). 
 Obtaining this construction of $\CJ$ from the complicated expression (\ref{CJ2expl-alt}) has {so far} resisted our attempts.

        \begin{figure}[htb]
  \centering
       \includegraphics[width=20pc]{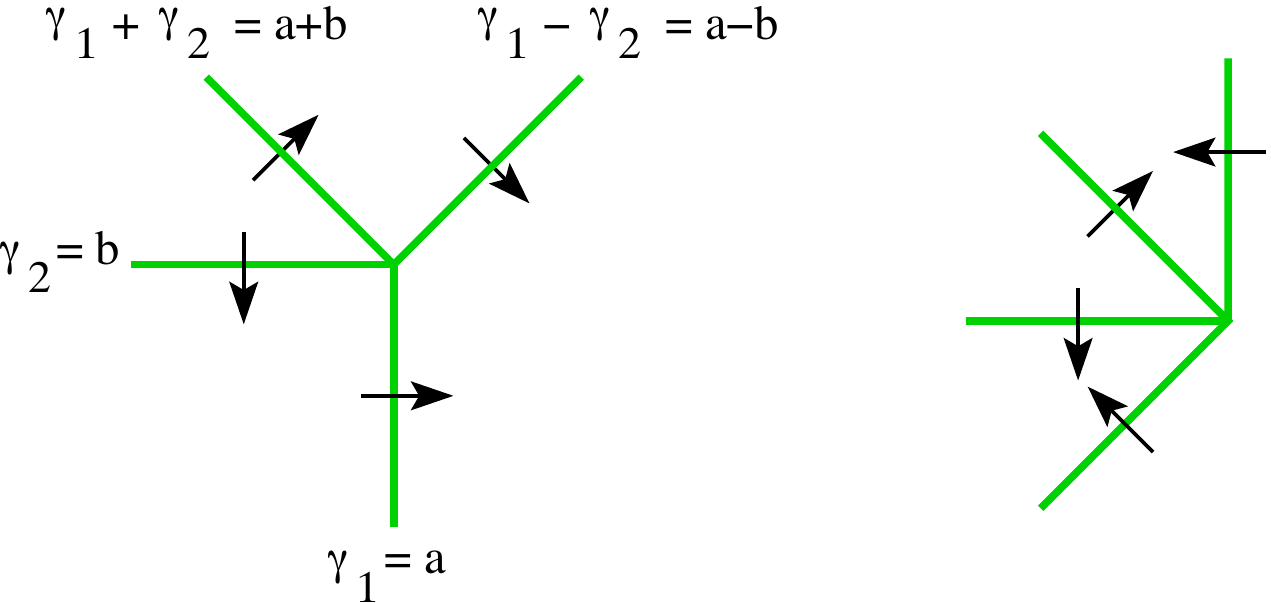}  \caption{\label{four-prong} Prescriptions of changes of $\CJ$ across the lines
       emanating from a four-prong vertex. This holds for any rotated configuration of those two types. }  \end{figure}
       
   
  \subsection{Geometry of the $B_2$ root space. Possible vanishing of $C_{\lambda\,\mu}^\nu$  and saturation property}
  \label{sec:rootspace}
   Recall that if $e_i$, $i=1,2$, are two orthonormal vectors, the two simple roots may be written 
  as $\Ga_1=e_1-e_2,\ \Ga_2=e_2$, while the fundamental weights are 
  $\omega_1=e_1=\Ga_1+\Ga_2,\ \omega_2=\oh(e_1+e_2)$.  See Fig. \ref{B2rootsystem}. 
  In what follows, weights will usually be specified by their coordinates in the $(\omega_1,\omega_2)$ basis
  (``Dynkin labels"). Occasionally we'll have to use the simple root basis (``Kac indices") or the $e_i$ basis. 
  
   \begin{figure}[htb]
  \centering
       \includegraphics[width=10pc]{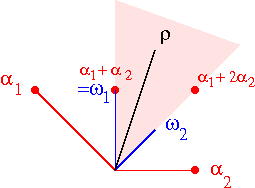}\caption{\label{B2rootsystem} Roots and weights of $B_2$: the positive roots in red, 
       the two fundamental weights in blue, and the Weyl vector in black. The shaded octant is the first (dominant) Weyl chamber.}
       \end{figure}

\bigskip       In the case of $B_2$, the compatibility condition
        $\sigma: =\lambda+\mu-\nu\in Q$ amounts to $\sigma_2=0\mod 2$.
          As an example, for $\lambda=\mu=\nu=\omega_2$ (the spinorial representation), $\sigma=\omega_2= \Ga_1/2 + \Ga_2 \notin Q$, and $C_{\lambda\,\mu}^\nu=0$ as this is a non-compatible triple. 
                        {By way of contrast, if  $\lambda=\mu=\nu=\omega_1$ (the vectorial representation),  the triple is compatible since 
           $\sigma=\omega_1 = \Ga_1+ \Ga_2 \in Q$,  but nevertheless we still have $C_{\lambda\,\mu}^\nu=0$. 
                Thus in both cases $C_{\lambda\,\mu}^\nu =0$, {\rm while} $C_{2\lambda\,2\mu}^{2\nu}=1\,.$

\bigskip
A semisimple algebra $\gog$ is said to satisfy the {\it saturation property} if $C_{N\lambda\, N\mu}^{N\nu}\ne 0$ implies $ C_{\lambda\, \mu}^{\nu}\ne 0$ whenever $N \in \N$ and $(\lambda, \mu, \nu)$ is a compatible triple.  The saturation property is proved for $A_n$ but fails for $B_n$, in particular for $B_2$,  as the example just given shows.\\

 \subsection{Berenstein--Zelevinsky parameters}
 \label{sec:BZ parameters}
 For the convenience of the reader, we reproduce here the result of Theorem 2.4 of \cite{BZ3}, specialized to the case of $B_2$.
  Introduce four real parameters $ t_0^{(0)},t_{-1}^{(1)}, t_0^{(1)},  t_1^{(1)}$, that express $\sigma=\lambda+\mu-\nu$
  in terms of the positive roots:
 $$ \sigma= ( t_{-1}^{(1)}- t_0^{(1)}+2 t_1^{(1)}) \Ga_1 +  t_0^{(0)} \Ga_2 +(  t_0^{(1)} -2 t_1^{(1)}) (\Ga_1+\Ga_2)
 +  t_1^{(1)} (\Ga_1+2\Ga_2) \,. $$
 Note that these parameters are linear combinations of the $\t_a$ associated to the (overcomplete) family of positive roots in (\ref{part-poly}).
 The $t$'s are   subject to the conditions
  \bea\nonumber 
  2 t_{-1}^{(1)} \ge  t_0^{(1)} \ge 2  t_1^{(1)} \ge 0 \quad&{\rm and }& \quad   t_0^{(0)}\ge 0\\
 \label{BZ-ineq-B2} 
  \sigma=\lambda+\mu-\nu &=&( t_1^{(1)}+ t_{-1}^{(1)} ) \Ga_1 +   (t_0^{(0)} + t_0^{(1)} )\Ga_2\\ 
\nonumber
  \lambda_1 \ge \max( t_1^{(1)},  t_0^{(1)}- t_{-1}^{(1)} ,  t_{-1}^{(1)}- t_0^{(0)})\ &{\rm and }&\    \lambda_2 \ge  t_0^{(0)}  \\
 \nonumber
   \mu_1 \ge \max(t_{-1}^{(1)}  +2 t_1^{(1)}-t_0^{(1)} , t_1^{(1)})\ &{\rm and }&\ 
		  \mu_2 \ge \max( t_0^{(0)}+2(  t_0^{(1)}- t_{-1}^{(1)}- t_1^{(1)})  ,  t_0^{(1)}-2 t_1^{(1)})\,. 
  \eea
  Then $C_{\lambda\,\mu}^{\nu}$ is equal to the number of integral solutions of (\ref{BZ-ineq-B2}).
 
  
  \subsection{Stretching. Parameter polytope}
 Under stretching, $P_{\lambda\,\mu}^\nu(s):=C_{s\lambda\,s\mu}^{s\nu}$  is in general
 not a polynomial of $s$ but a {\it quasi-polynomial}.\\
 For example for $\lambda=(5,6),\ \mu=(3,4),\ \nu=(5,6)$, 
 we find 
\be\label{Pol563456}P_{\lambda\,\mu}^\nu(s)= 6 s^2 +\frac{7}{2} s +\inv{4}(3+(-1)^s)\,. \ee
The 2-dimensional parameter polytope defined by (\ref{BZ-ineq-B2})  
 is therefore in general {\it not} an integral polytope.
 
 In the previous example where $\lambda=(5,6),\ \mu=(3,4),\ \nu=(5,6)$, 
 after eliminating\footnote{\label{refnote}We must make sure in the final counting however that the eliminated parameters are also integers.} the parameters $t_{-1}^{(1)}$ and $t_0^{(1)}$ through the 2 equalities 
 $\lambda+\mu-\nu=\cdots$, we find the polytope (here a polygon) in the $(t_0^{(0)},t_1^{(1)})$ plane
 defined by the inequalities 
 $$ \{0\leq t_0^{(0)}<6\land 0\leq t_1^{(1)}\leq 3\land t_0^{(0)}+3\geq 2 t_1^{(1)}\land
   3\leq t_0^{(0)}+2 t_1^{(1)}\leq 7\land t_0^{(0)}+t_1^{(1)}\leq 5\}$$
 depicted in Fig. \ref{ParamPoly}(a). It has $C_{\lambda \mu}^\nu =10$ integral points but 
 its corners are not integral,  and  its Ehrhart quasi-polynomial is
 given in (\ref{Pol563456}). Note that its  (relative) 
 ``volume,'' here an area, is $V=6$, in accordance with  the computation 
 of $\CJ$.
   
 In contrast,  for $\lambda=(5,6),\ \mu=(3,4),\ \nu=(6,4)$, we find $C_{\lambda\,\mu}^\nu =10$, $V=11/2$,
  \be\label{Pol563464}P_{\lambda\,\mu}^\nu(s)= \frac{11}{2} s^2 +\frac{7}{2} s +1 \,, \ee
 and the parameter polygon defined by
 $$ \{2\leq t_0^{(0)}\leq 6\land 0\leq t_1^{(1)}\leq 3\land t_0^{(0)}+2\geq 2
   t_1^{(1)}\land 4\leq t_0^{(0)}+2 t_1^{(1)}\leq 8\land t_0^{(0)}+t_1^{(1)}\leq 6\}$$
  is shown in Fig. \ref{ParamPoly}(b).  Note that the polygon is again not integral, although its Ehrhart quasi-polynomial is the genuine polynomial (\ref{Pol563464}). 
  
There are also cases where the polygon is integral.  For example, still with  $\lambda=(5,6),\ \mu=(3,4)$, and now
 $\nu=(2,10)$, we get $C_{\lambda\,\mu}^\nu =8$  and
  \be\label{Pol5634210}P_{\lambda\,\mu}^\nu(s)= \frac{7}{2} s^2 +\frac{7}{2} s +1 \,, \ee
 see Fig. \ref{ParamPoly}(c).
 
Note that in the three previous cases, $P_{\lambda\,\mu}^\nu(-1)$ gives the number of internal points, as it should \cite{Stanley}:
 respectively, 3, 3 and 1.

\bigskip

 In order to make a few simple comments of a geometrical nature, we assume in the rest of this subsection that 
 the sub-leading coefficient of $P(s)$ is constant.  We have found no counter-example of this property for $B_2$ BZ polygons, but we do not offer a proof.}
{In other words we assume that (non-degenerate) $B_2$ stretching polynomials read  $P (s) = V s^2 + L/2 s + (p + (-1)^s (1-p))$ where $V$ is the  (relative) area of the parameter polygon, $L$ is the (relative) length of its boundary and $p$ is some scalar. For a genuine polynomial --- \ie not quasi --- one has $p = 1$.
The polytope being closed and convex we know that $P(0)=1$. Let $C:= C_{\lambda \mu}^{\nu} = P(1)$ be the LR coefficient, i.e. the total number of integer points of the polytope.  Then $i := P(-1)$ is the number of interior points (by Ehrhart--Macdonald reciprocity), and $b := C-i$ is the number of integer points belonging to the boundary. Evaluation of $P$ at $+1$ and $-1$ gives $C + i = 2(V + (2p-1))$ and $C - i = L$; together these two equations imply 
 \be\label{PickDeviation}
 2C -2V - L = 2 (2p-1).
 \ee
Moreover the two relations $C = b + i$ and $C=L+i$ imply $L=b$. All these relations can be checked in the three examples above (see Fig. \ref{ParamPoly}).

Notice that the polytope is integral  only if $p=1$, i.e. if  $2 C - 2 V - L = 2$; this is what happens in the third example above, where $C=8$.
In general, $p$ may be read off from (\ref{PickDeviation}). 
 For instance in the first example  where $V=6$, $L=7$ and  $C=10$, we find $2 C - 2 V - L = 1$ and $p=3/4$, in agreement with the expression of the stretching polynomial.  When the polytope is integral, Pick's theorem, written $V = i + b/2 -1$, applies, and this is of course {equivalent to (\ref{PickDeviation}) with $p = 1$.

 \begin{figure}[htb]
  \centering
       \includegraphics[width=10pc]{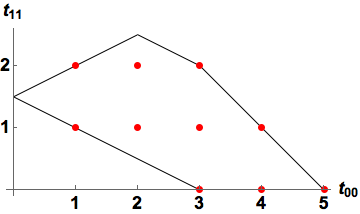}\qquad  \includegraphics[width=10pc]{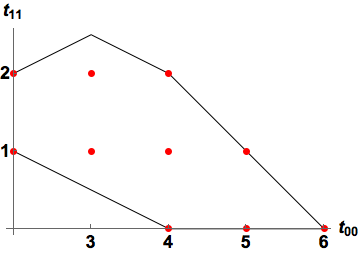}
      \qquad  \includegraphics[width=10pc]{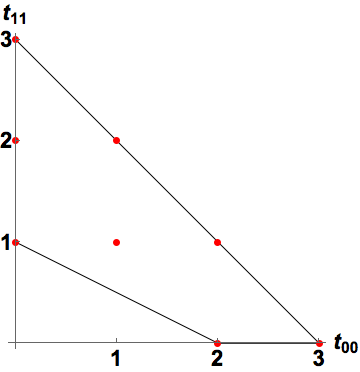}\\[6pt]   
  \   (a)\qquad  \qquad \qquad  \qquad   \qquad\quad(b) \qquad  \qquad \qquad \qquad \qquad  \qquad(c)
       \caption{\label{ParamPoly} The parameter polytope in the $t_0^{(0)}-t_1^{(1)}$ plane for 
   $\lambda=(5,6),\ \mu=(3,4)$ and      (a) $\nu=(5,6)$; (b) $\nu=(6,4)$; 
       (c)  $\nu=(2,10).$ }\end{figure}

\normalcolor
 Finally there are  cases where the polygon degenerates, either to a point (whenever $C_{\lambda\,\mu}^{\nu}=1$), or to a segment. 
 In the $A_r$ case, $C_{\lambda\,\mu}^{\nu}=1$  implies
 $\forall s, \ C_{s\lambda\,s\mu}^{s\nu}=1$, as proved in \cite{KTW01}.   Whether this holds true for other cases like $B_r$ seems to
 be still an open question. 
The polytope may also degenerate to a segment, as occurs for example when we keep $\lambda=(5,6),\ \mu=(3,4)$, but take $\nu=(0,10)$. Then we find $C_{\lambda\,\mu}^{\nu}=3$; the segment has length $2$,  and upon dilation $P_{\lambda\,\mu}^{\nu}(s)=2s+1$.


{ 
  \subsection{Volumes and multiplicities in practice}
 
  \subsubsection{Determination of multiplicities in the $B_2$ case}

In order to determine the 
generalized Littlewood-Richardson coefficients in the $B_2$ case, one can use the Racah--Speiser algorithm \cite{Racah-Speiser}
which works for any semisimple Lie algebra (even for affine Lie algebras),  or equivalently Klimyk's formula \cite{Klim}, which is implemented in several computer algebra packages such as LiE \cite{LiE}.
Another possibility, since the Kostant partition function is known for $B_2$ (see  \cite{Tarski},  \cite{Capparelli}), is to use the Steinberg formula \cite{Ste}.
A third possibility is to use Berenstein-Zelevinsky polytopes, as explained in sections \ref{sec:BZpol} and \ref{sec:BZ parameters}.
We implemented these three methods\footnote{In the $A_r$ cases we prefer to use our own version (O-blades, see \cite{CZ3}) of an algorithm using honeycombs, because it has an easy interpretation and because it is fast.} in Mathematica \cite{Mathematica}, which was also used for most formal manipulations done in this paper (and also for graphics).

 
 \subsubsection{The many ways to compute the volume of a BZ polytope}
\label{many-ways}
Let $(\lambda, \mu, \nu)$ be a  given  compatible triple of highest weights of $\gog$.
 We are interested in the (relative) volume $V$ of the associated BZ polytope.

There are --- at least --- four ways to compute this:\\
1) One can use the volume function ${\mathcal J}_r (\lambda, \mu, \nu)$ when it is explicitly known, as is the case for $A_r$, with $r=1,2,3,4$ (see \cite{Z1} and sect. 4 of \cite{CZ1}), and for $B_2$ via the expression (\ref{CJ2expl-alt}).\\
2) One can use  (\ref{In-LRp}) (also formula (36) of \cite{CZ1}) that expresses  $V$ in terms of a finite number of multiplicities (LR coefficients) and of the constants $\hat\c_{\kappa}$ (defined in sect. \ref{sectJ-LR}) associated to $\gog$.\\
3) One can determine the stretching quasi-polynomial defined by the triple, by calculating stretched multiplicities $C_{s\lambda\, s\mu}^{s\nu}$
up to some $s$ of the order of $\varpi \times d_r$ where $\varpi$ is the period of the quasi-polynomial (not more than $2$ for classical Lie algebras), and $d_r$ is the degree (see Table 1 in sect. \ref{sectVol}). Then $V$ is the coefficient of the leading-order term.\\
4) When the BZ-polytope is explicitly defined in terms of appropriate parameters (for instance the BZ-parameters described previously for $B_2$), one can compute its volume by integrating the constant function $1$ on the polytope,  possibly after eliminating redundant parameters.  When using the $\t_a$ parameters of (\ref{part-poly}),  this method will compute the Euclidean volume, so to obtain the relative volume one must divide by the covolume $c_\gog = \delta_r^{1/2}$, see Table 1.

For illustration, let us consider the following triple for $B_2$: $\lambda=(4,7), \mu=(5,3), \nu=(2,4)$ {in the basis of fundamental weights.  Equivalently in the $e_i$ basis, $\lambda = (15/2, 7/2), \mu = (13/2, 3/2), \nu = (4,2)$.}  This triple has multiplicity $5$. In what follows, coordinates of weights are written in the fundamental weight basis. \\
1) Direct evaluation of  (\ref{CJ2expl-alt}) with the above arguments gives  $V=7/4$.\\
2) The set $\hat K$ contains only one element,  $\kappa=(0,1)$, with $\c_{\kappa} = 1/4$.  The non-zero contribution to  (\ref{In-LRp}) comes from the following weights (with multiplicities) that enter the decomposition of the tensor product of {\footnotesize $(3,6)=(4,7)-(1,1)$} and {\footnotesize $(4,2)=(5,3)-(1,1)$}:  {\footnotesize $\{(0, 4), 1\}, \{(1,2), 1\}, \{(1,4), 3\}, \{(2,2), 2\}$}, so that one obtains $V=(1+1+3+2)/4 = 7/4$.
 More generally, for $\nu$ deep enough (see (\ref{deep-enough})), one finds 
$$ \CJ(\lambda,\mu; \nu)= \inv{4} \sum_k C_{(\lambda-\rho) (\mu-\rho) }^{\nu-k}\qquad 
\mathrm{with}\quad k\in \{ (2,0),\ (1,2),\ (1,0),\ (0,2)\}\,.$$\\
3)  For scaling factors $s=0, 2, 4$ the multiplicities are respectively $1, 13, 39$. For scaling factors $s=1, 3, 5$ the multiplicities are respectively $5, 24, 57$.
The quasi-polynomial reads $\frac{7 s^2}{4}+\frac{5 s}{2}+1$ when $s = 0 \,\rm{mod}\,2$ and  $\frac{7 s^2}{4}+\frac{5 s}{2}+\frac{3}{4}$ when $s = 1 \,\rm{mod}\,2$. As expected, the leading coefficient of both polynomials is $V=7/4$.
\\
4) In the ${t_0^{(0)}, t_1^{(1)}}$ plane, this BZ polytope is defined by the inequalities: \\
{\footnotesize $( t_1^{(1)} = 2 \textrm{ and }  t_0^{(0)} = 6) \textrm{ or } (2 <  t_1^{(1)} \leq 7/2 \textrm{ and } 
   10 - 2  t_1^{(1)} \leq  t_0^{(0)} \leq 8 -  t_1^{(1)}) \textrm{ or }  (7/2 <  t_1^{(1)} \leq 4 \textrm{ and } 
   3 \leq  t_0^{(0)} \leq 8 -  t_1^{(1)})$}. \\ It is displayed below (Fig. \ref{polytope475324}); note that it is not an integral polygon. Its area ($V=7/4$) can be readily calculated.\\
   
 \begin{figure}[htb]
  \centering
\includegraphics[width=10pc]{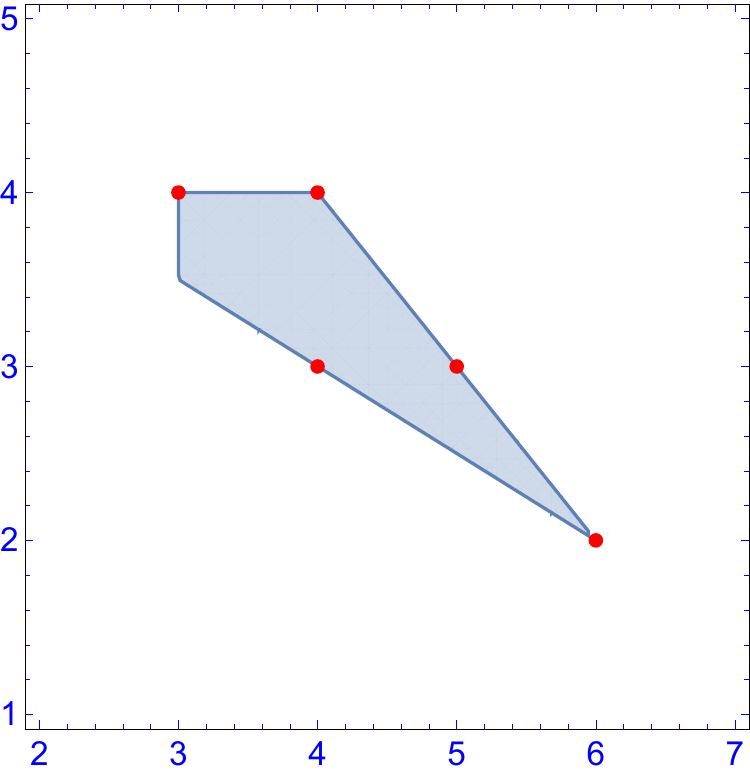}
 \caption{\label{polytope475324} The BZ polytope in the $(t_0^{(0)}, t_1^{(1)})$ plane for $\lambda=(4,7),\ \mu=(5,3),\ \nu=(2,4)$.}\end{figure}

\subsubsection{Determination of the coefficients $\c_\kappa$ and $\hat \c_\kappa$} 
As explained in sect. \ref{sectJ-LR}, 
 the formula (\ref{kissinger}) can be used to determine the constants $\c_\kappa$ and $\hat \c_\kappa$.   
Take for example  $\kappa = (0,0)$ in $B_2$; the triple $(\rho, \rho, \kappa + \rho)$ used in the first equation (\ref{kissinger}) is never compatible. 
One computes the multiplicity of the scaled triples $(s(1,1),s(1,1),s(1,1))$, which is $0$ when $s$ is odd (in particular for $s=1$), and equal to $(1, 4, 10, \ldots)$ when $s = 0, 2, 4, \ldots$; the quasi-polynomial is $0$ when $s$ is odd, and $\frac{3 s^2}{8}+\frac{3 s}{4}+1$ when $s$ is even. The leading coefficient obtained for $s$ even gives  $\c_{(0,0)}=3/8$, as desired.

Likewise for $B_3$, the coefficient $\c_\kappa$ associated with $\kappa = (0,0,0)$ can be determined from the scaled triples $(s(1,1,1),s(1,1,1),s(1,1,1))$. Here the period of the stretching quasi-polynomial is $4$ (a result which is not in contradiction with known theorems, since the triple is not compatible).  The LR quasi-polynomial vanishes for odd $s$, whereas when $s=0 \, \rm{mod} \, 4$ or $s=2 \, \rm{mod} \, 4$, one gets respectively 
$\frac{241 s^6}{3072}+\frac{241 s^5}{512}+\frac{4165 s^4}{3072}+\frac{19
   s^3}{8}+\frac{523 s^2}{192}+2 s+1$ or 
   $\frac{241 s^6}{3072}+\frac{241 s^5}{512}+\frac{4165 s^4}{3072}+\frac{281
   s^3}{128}+\frac{839 s^2}{384}+\frac{19 s}{16}+\frac{35}{64}$.
Therefore $\c_{(0,0,0)} =  241/3072 = 7230/92160$ as given in sect. \ref{sectJ-LR}. 
  Notice that here one has to use scaling factors up to $s=6\times4=24$ and $s=6\times 4+2=26$, to obtain the two non-trivial polynomial determinations of this quasi-polynomial.
  }


 \section*{Acknowledgements} 
We acknowledge fruitful discussions with P.~Di Francesco, P.~Etingof, V.~Gorin, R.~Kedem, S.~Sam and  M.~Vergne.
The work of Colin McSwiggen is partially supported by the National Science Foundation under Grant No.~DMS 1714187, as well as by the Chateaubriand Fellowship of the Embassy of France in the United States.

\section*{Appendix: Covolumes}

 Let us sketch how the covolume $c_\gog$ may be determined for the affine lattice $\Lambda$ of integer points of $\mathrm{aff}(H_{\lambda \mu}^\nu)$, in the case that $\dim H_{\lambda \mu}^\nu = d_r$ so that $\mathrm{aff}(H_{\lambda \mu}^\nu) = \mathrm{aff}(\mathrm{Part}_\gog(\sigma))$, see (\ref{part-poly}).
This affine span has codimension $r$ in the Euclidean space $\R^{N_r}$ generated as the formal span of the positive roots, which are taken to form a canonical basis, i.e. $\langle \Ga_a, \Ga_b\rangle = \delta_{ab}$, $a,b=1,\cdots, N_r$. 
Since the covolume of $\Lambda$
 is independent of the value of $\sigma$, we may take $\sigma=0$. Then $\Lambda$ is the lattice of integer combinations
 of the positive roots $\Ga_a$ subject to the condition 
 \be \label{cond-r}\sum_{a=1}^{N_r} \t_a \Ga_a =0\in  \mathfrak{t}^*.\ee
  Suppose  that the first $r$ of the $\Ga$'s are the simple roots, and denote by $A$ the $(N_r-r) \times r$ 
  matrix that expresses the non-simple roots in terms of the simple ones: 
  $$\Ga_a=\sum_{i=1}^r A_{ai}\Ga_i, \quad a=r+1,\hdots, N_r.$$
Then eliminate the parameters $u_i$ for $i=1,\hdots,r$ in the condition (\ref{cond-r}), using the relation
$u_i= - \sum_{a=r+1}^{N_r}  u_a A_{ai}$.  The points of $\Lambda$ can then all be written in the form
$\sum_{a=r+1}^{N_r}  u_a (-\sum A_{ai} \Ga_i +\Ga_a)$. {\it Under the assumption} that the parallelepiped spanned by the $N_r-r$ vectors $ w_a= (-\sum A_{ai} \Ga_i +\Ga_a)$ is a fundamental domain of $\Lambda$, the covolume of $\Lambda$ is equal to the volume of this parallelepiped, which we can compute as follows.
 The Gram matrix of the vectors $w_a$ in the Euclidean space  $\R^{N_r}$ reads
$$ G_{ab}=\langle w_a, w_b\rangle = \delta_{ab}+ \sum_{i=1}^r A_{ai}A_{bi}\,.$$
The covolume is then the square root of the determinant $\delta_r$ of $G$, $c_\gog=\delta_r^\oh$. 
We have computed these determinants for the classical algebras 
$A_r,\ B_r,\ C_r$ and $D_r$ up to $r=8$,  and the values of $\delta_r$ listed in Table 1 are extrapolations of these numbers. The values for the exceptional algebras  are also included in Table 1.
A last observation is that for  all the simple algebras, a single formula encompasses the expressions found in the table:
\be\label{last-formula}  
\delta_r=  \frac{(h^\vee)^r} {\det C} \,  \prod_{i=1}^r \frac{\langle \pmb{\theta}, \pmb{\theta} \rangle}{\langle\Ga_i,\Ga_i\rangle}
\,,\ee
where $C$ is the Cartan matrix, $h^\vee$ is the dual Coxeter number, $\pmb{\theta}$ is any long root, and the product runs over the simple roots.


 \end{document}